\newif\ifIEEE
\IEEEtrue
\ifIEEE
   \documentclass[journal,final,letterpaper]{IEEEtran}
   \newcommand{\bibauthor}[1]{#1}
   \newcommand{\bibpaper}[1]{``#1''}
   \newcommand{\Footnotetext}[2]
   {
      \begin{figure}[!b]
      \footnotesize\vspace{-3ex}\hrulefill\hfill
      \makebox[0em]{}\hfill\makebox[0em]{}\par${}^{#1}$ #2\vspace{-.6ex}
      \end{figure}
      \addtocounter{figure}{0}
   }
\else
   \documentclass[12pt]{article}
   \usepackage{amsthm}
   \newcommand{\bibauthor}[1]{\textsc{#1}}
   \newcommand{\bibpaper}[1]{\textsl{#1}}
   \newenvironment{IEEEkeywords}{\begin{small}%
                                  \textbf{Index Terms} ---}{\end{small}}
\fi
\usepackage{amsfonts,bm}
\usepackage{amsthm}
\usepackage{eepic}
\newcommand{\bibbook}[1]{\textit{#1}}

\newcommand{\bibperiodical}[1]{\textit{#1}}

\newtheorem{theorem}{Theorem}
\newtheorem{proposition}[theorem]{Proposition}

\theoremstyle{remark}

\theoremstyle{definition}
\newtheorem{example}{Example}

\renewcommand{\mathbf}[1]{{\bm{#1}}}     
\newenvironment{myalgorithm}{%
        \begin{minipage}{\columnwidth}\vspace{1.8ex}
        \makebox[0ex]{}\hrulefill\makebox[0ex]{}\\*%
             }{%
             \makebox[0ex]{}\hrulefill\makebox[0ex]{}\end{minipage}}
\newenvironment{clause}{%
        \begin{list}{}{\setlength{\topsep}{0ex}%
             \settowidth{\leftmargin}{\quad}%
             \setlength{\itemsep}{1.0ex}%
             \setlength{\parsep}{-.5ex}}}{\end{list}}
\newcommand{\GF}{{\mathrm{GF}}}
\newcommand{\bldzero}{{\mathbf{0}}}
\newcommand{\bldalpha}{{\mathbf{\alpha}}}
\newcommand{\bldbeta}{{\mathbf{\beta}}}
\newcommand{\bldc}{{\mathbf{c}}}
\newcommand{\blde}{{\mathbf{e}}}
\newcommand{\blds}{{\mathbf{s}}}
\newcommand{\bldu}{{\mathbf{u}}}
\newcommand{\bldw}{{\mathbf{w}}}
\newcommand{\bldx}{{\mathbf{x}}}
\newcommand{\bldy}{{\mathbf{y}}}
\newcommand{\code}{{\mathcal{C}}}
\newcommand{\Code}{{\mathsf{C}}}
\newcommand{\Ber}{{\mathrm{Ber}}}
\newcommand{\Integers}{{\mathbb{Z}}}
\newcommand{\Int}[1]{{\left[{#1}\right\rangle}}
\newcommand{\distance}{{\mathsf{d}}}
\newcommand{\weight}{{\mathsf{w}}}
\newcommand{\encoder}{{\mathcal{E}}}
\newcommand{\Encoder}{{\mathsf{E}}}
\newcommand{\decoder}{{\mathcal{D}}}
\newcommand{\Decoder}{{\mathsf{D}}}
\newcommand{\Lee}{{\mathcal{L}}}
\newcommand{\Hamming}{{\mathcal{H}}}
\newcommand{\Sphere}{{\mathcal{S}}}
\newcommand{\Volume}{V}
\newcommand{\failure}{{\mathrm{``e"}}}
\newcommand{\Mod}{{\scriptstyle \mathrm{MOD}}}
\newcommand{\transpose}{{\mathsf{T}}}
\newcommand{\mod}{{\mathrm{mod}}}
\newcommand{\threshold}{\vartheta}
\newcommand{\tn}{{\tilde{n}}}
\newcommand{\Title}{Fault-Tolerant Dot-Product Engines}
\newcommand{\Name}{Ron M. Roth}
\newcommand{\Address}{Computer Science Department,
                      Technion,
                      Haifa 3200003, Israel}
\newcommand{\Email}{ronny@cs.technion.ac.il}
\newcommand{\Thnx}{\Name\ is with \Address. Email: \Email}
\newcommand{\Grant}{This work was done in part while visiting
            Hewlett Packard Laboratories, 1501 Page Mill Road,
            Palo Alto, CA 94304.}

\begin{document}
\ifIEEE
   \title{\Title}
      \author{\IEEEauthorblockN{\Name}\\
              \IEEEauthorblockA{\Address\\ \Email}
      }
\else
   \title{\textbf{\Title}\thanks{\Grant}}
   \author{\textsc{\Name}\thanks{\Thnx}}
   \date{}
\fi
\maketitle
\ifIEEE  \thispagestyle{empty}  \fi

\begin{abstract}
Coding schemes are presented that provide the ability
to correct and detect computational errors while using dot-product
engines for integer vector--matrix multiplication.
Both the $L_1$-metric and the Hamming metric are considered.
\end{abstract}


\ifIEEE
   \Footnotetext{\quad}{\Grant}
\fi

\section{Introduction}
\label{sec:introduction}

We consider the following computational model.
For an integer $q \ge 2$, let~$\Sigma_q$ denote the subset
$\Int{q} = \{ 0, 1, \ldots, q{-}1 \}$ of
the integer set~$\Integers$.  Also, let~$\ell$ and~$n$ be fixed
positive integers. A \emph{dot-product engine}
(in short, DPE) is a device which accepts as input
an $\ell \times n$ matrix
$A = (a_{i,j})_{i \in \Int{\ell}, j \in \Int{n}}$ over~$\Sigma_q$
and a row vector $\bldu = (u_i)_{i \in \Int{\ell}} \in \Sigma_q^\ell$,
and computes the vector--matrix product $\bldc = \bldu A$, with
addition and multiplication carried out over~$\Integers$. Thus,
$\bldc = (c_j)_{j \in \Int{n}}$ is an integer vector in~$\Integers^n$
(more specifically, over $\Sigma_{\ell(q-1)^2+1}^n$).
In the applications of interest, the matrix~$A$ is modified
much less frequently than the input vector~$\bldu$
(in some applications, the matrix~$A$ is determined once
and then remains fixed, in which case only~$\bldu$ is seen as input).
Typically, the alphabet size\footnote{%
One can consider the broader problem where the matrix~$A$ and
the vector~$\bldu$ are over different integer alphabets.
Yet, for the sake of simplicity, we will assume hereafter that
those alphabets are the same.
It is primarily the alphabet of the matrix that will affect
the coding schemes that will be presented in this work.}
$q$ is a power of~$2$.

In recent proposals of nanoscale implementations of a DPE,
the matrix~$A$ is realized as a crossbar array consisting of $\ell$~row
conductors, $n$~columns conductors, and programmable nanoscale
resistors (e.g., memristors) at the junctions, with the resistor
at the junction~$(i,j)$ set to have conductance, $G_{i,j}$,
that is proportional to $a_{i,j}$.
Each entry~$u_i$ of~$\bldu$ is fed into a digital-to-analog converter
(DAC) to produce a voltage level that is proportional to~$u_i$.
The product, $\bldu A$, is then computed by reading the currents at
the (grounded) column conductors, after being fed into
analog-to-digital converters (ADCs);
see Figure~\ref{fig:crossbar}.
For early implementations and applications of DPE's,
as well as recent ones, see, for example,~%
\cite{BSBLJ},
\cite{HSKGDGLGWY},
\cite{KMML},
and~\cite{SNMBSHWS}.
\begin{figure}[bt]

\newlength{\figunit}
\setlength{\figunit}{0.23ex}

\newcommand{\BigBullet}{\thicklines
    \put(000,000){\circle{2.4}}
    \put(000,000){\circle{1.2}}
}

\newcommand{\Bullet}{\thicklines
    \put(000,000){\circle{1.6}}
    \put(000,000){\circle{0.8}}
}

\newcommand{\ground}{
    \put(-03,000){\line(1,0){6}}
    \put(-02,-01){\line(1,0){4}}
    \put(-01,-02){\line(1,0){2}}
}

\newcommand{\dmemristor}{
    \thicklines
    \put(000,000){
        \put(0.5,-.5){
            \put(-11,000){\line(1,-1){2}}
            \put(-09,-02){\line(1,1){1}}
            \put(-08,-01){\line(1,-1){2}}
            \put(-06,-03){\line(-1,-1){2}}
            \put(-08,-05){\line(1,-1){2}}
            \put(-06,-07){\line(1,1){2}}
            \put(-04,-05){\line(1,-1){2}}
            \put(-02,-07){\line(-1,-1){1}}
            \put(-03,-08){\line(1,-1){2}}
        }

    }
    \put(-11,000){\Bullet}
    \put(000,-011){\Bullet}
}

\newsavebox{\Bullets}
\sbox{\Bullets}{
    \setlength{\unitlength}{\figunit}
    \put(000,000){
        \multiput(-18,000)(0,18){6}{
            \put(000,000){\BigBullet}
        }

        \multiput(000,-18)(18,0){8}{
            \put(000,000){\BigBullet}
        }
    }
}

\newsavebox{\Grid}
\sbox{\Grid}{
    \setlength{\unitlength}{\figunit}
    \put(000,000){
        \multiput(000,-18)(18,0){8}{
            \put(000,000){\line(0,1){117}}
        }

    }

}

\setlength{\unitlength}{\figunit}
\begin{center}
\begin{picture}(209,191)(-70,-74)

\put(000,000){\usebox{\Grid}}
%
%
\put(-28,000){
        \multiput(000,000)(000,018){6}{
                {\thicklines
                    \multiput(000,-07)(-14,0){2}{\line(0,1){14}}
                    \multiput(000,-07)(0,14){2}{\line(-1,0){14}}
                    \put(-07,000){\makebox(0,0){\tiny DAC}}
                }
                \put(000,000){\vector(1,0){10}}
                \put(-24,000){\vector(1,0){10}}
        }
        \put(-30,000){
                \put(000,090){\makebox(0,0){$u_0$}}
                \put(000,072){\makebox(0,0){$u_1$}}
                \put(000,056.5){\makebox(0,0){$\vdots$}}
                \put(000,036){\makebox(0,0){$u_i$}}
                \put(000,020.5){\makebox(0,0){$\vdots$}}
                \put(-04,000){\makebox(0,0){$u_{\ell-1}$}}
        }
}
%
%
\put(000,-22){
        \multiput(000,000)(018,000){8}{
                {\thicklines
                    \put(000,-07){\circle{10}}
                    \put(000,-07){\makebox(0,0){\tiny A}}
                    \multiput(-07,-22)(0,-14){2}{\line(1,0){14}}
                    \multiput(-07,-22)(14,0){2}{\line(0,-1){14}}
                    \put(000,-29){\makebox(0,0){\tiny ADC}}
                }
                \put(005,-07){\line(1,0){3}}
                \put(008,-07){\line(0,-1){09}}
                \put(008,-16){\ground}
                \put(000,004){\line(0,-1){06}}
                \put(000,-12){\vector(0,-1){10}}
                \put(000,-36){\vector(0,-1){10}}
        }
        \put(000,-52){
                \put(000,000){\makebox(0,0){$c_0$}}
                \put(018,000){\makebox(0,0){$c_1$}}
                \put(045,000){\makebox(0,0){$\cdots$}}
                \put(072,000){\makebox(0,0){$c_j$}}
                \put(099,000){\makebox(0,0){$\cdots$}}
                \put(126,000){\makebox(0,0){$c_{n-1}$}}
        }
}

\multiput(-18,000)(000,018){6}{
    \put(000,000){\line(1,0){153}}
    \multiput(018,000)(018,000){8}{\dmemristor}
}

\put(064,025){\makebox(0,0){\tiny $G_{i,j}$}}
\put(000,000){\usebox{\Bullets}}

\end{picture}
\end{center}
\caption{Schematic diagram of a DPE implementation
of the computation $\bldu \mapsto \bldc = \bldu A$ using
an $\ell \times n$ crossbar array of memristors.
The conductance~$G_{i,j}$ of the memristor at each junction $(i,j)$
is proportional to $a_{i,j}$. The circles marked ``A''
represent analog current measuring devices
(such as transimpedance amplifiers). The current measurements can be
carried out in parallel (as shown), or serially, column-by-column,
using only one measuring device.}
\label{fig:crossbar}
\end{figure}
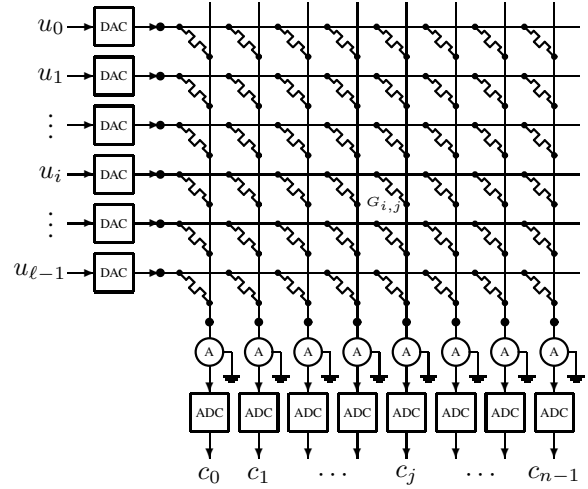

Inaccuracies while programming the resistors in the crossbar
and noise while reading the currents are examples of factors that
can affect the accuracy of the computation. Specifically,
the actually-read row vector,
$\bldy = (y_j)_{j \in \Int{n}} \in \Integers^n$, may differ from
the correct vector, $\bldc = \bldu A$.
The error vector is defined as the following
vector in~$\Integers^n$:
\[
\blde = (e_j)_{j \in \Int{n}} = \bldy - \bldu A \; .
\]
In such context of errors, we find it natural to define
the number of errors to be the $L_1$-norm of~$\blde$:
\[
\| \blde \| = \| \blde \|_1 = \sum_{j \in \Int{n}} |e_j| \; .
\]
In our case---where~$\blde$ is an integer---this norm is also
referred to as the Manhattan weight of~$\blde$, which equals
the Manhattan distance between $\bldu A$ and~$\bldy$.

Another source of computational errors could be
junctions in the crossbar becoming shorted
due to faults in the programming process\footnote{%
Shorts could also result from manufacturing defects, although
conceivably these can be detected before the DPE is put into
operation.}.
In this case, the current read in the shorted column will be above
some compliance limit (``$\infty$''), which, in turn, will flag
the respective entry in~$\bldy$ as ``unavailable''
or as an \emph{erasure}.
The $L_1$-metric has been studied quite extensively
in the coding literature,
along with its finite-field Lee-metric variant:
see~\cite[Ch.~9]{Berlekamp} and~\cite[Ch.~10]{Roth})
(and Subsection~\ref{sec:Lee} below).

In the other extreme, a junction in the array may become
non-programmable or get stuck at an open state, in which cases
the respective entry in~$\bldy$ could be off the correct value by
as much as $\pm (q{-}1)^2$. Such errors could be
counted through their contribution to the $L_1$-norm
of the error vector. Alternatively, if this type of errors
is predominant, one could consider the Hamming metric instead,
whereby the figure of merit is
the Hamming weight of~$\blde$, equaling the number of positions
in which~$\bldy$ differs from $\bldu A$ (disregarding the extent
at which the values of the respective entries actually differ\footnote{%
Yet we will also consider a more general setting,
where that difference is bounded by some prescribed constant.}).
This Hamming metric is suitable for handling erasures as well.

In this work, we propose methods for using the DPE computational
power to self-protect the computations against errors.
The first $k \; (< n)$ entries in $\bldc = \bldu A$ will carry
the (ordinary) result of the computation of interest, while
the remaining $n-k$ entries of~$\bldc$ will contain
redundancy symbols, which can be used to detect or
correct computational errors, assuming that the number of the latter
(counted with respect to either the $L_1$-metric or
the Hamming metric) is bounded from above
by some design parameter. Specifically,
the programmed $\ell \times n$ matrix $A$ will have the structure
\[
A = \left( A' \;|\; A'' \right) \; ,
\]
where~$A'$ is an $\ell \times k$ matrix over~$\Sigma_q$ consisting of
the first $k$~columns of~$A$, and~$A''$ consists of
the remaining $n-k$~columns;
the computed output row vector for an input vector
$\bldu \in \Sigma_q^\ell$ will then be
$\bldc = \left( \bldc' \;|\; \bldc'' \right)$,
where the $k$-prefix $\bldc' = \bldu A' \; (\in \Integers^k)$ represents
the target computation while the $(n{-}k)$-suffix
$\bldc'' = \bldu A'' \; (\in \Integers^{n-k})$ is
the redundancy part.
In this setting, $A'$ and~$\bldu$ are the actual inputs, and~$A''$
will need to be computed from~$A'$, e.g., by a dedicated circuitry,
prior to---or while---programming~$A'$ and~$A''$ into
the crossbar array
(yet recall that it is expected that~$A'$ will be modified much less
frequently than~$\bldu$).
The error decoding mechanism will be implemented by
dedicated circuitry too.
Clearly, we will aim at minimizing $n-k$ given
the designed error correction capability.

\begin{example}
\label{ex:paritycode}
Let us consider the simplest case where we would
like to be able to only detect one $L_1$-metric error.
In this case, we select
$n = k+1$ and let the $\ell \times 1$ matrix
$A'' = (a_{i,k})_{i \in \Int{\ell}}$ be obtained from
$A' = (a_{i,j})_{i \in \Int{\ell},j \in \Int{k}}$ by
\[
a_{i,k} = \Bigl( \sum_{j \in \Int{k}} a_{i,j} \Bigr) \; \Mod \; 2 \; ,
\quad i \in \Int{\ell} \; ,
\]
where ``$\Mod$'' stands for (the binary operation of) remaindering;
thus, the entries of~$A''$ are in fact over~$\Sigma_2$,
and the sum of entries along each row of~$A$ is even.
It follows by linearity that the sum of entries
of (an error-free) $\bldc = (c_j)_{j \in \Int{n}} = \bldu A$
must be even.
On the other hand, if $\blde \in \Integers^n$ is an error vector
with $\| \blde \| = 1$
then the sum of entries of $\bldy = \bldc + \blde$ will be odd.\qed
\end{example}

Observe that the contents of~$A''$ depends on~$A'$, but
should \emph{not} depend on~$\bldu$. In particular, $A''$
should be set so that the specified error correction--detection
capabilities hold when~$\bldu$ is taken to be a unit vector.
Thus, for every row index~$i$,
the set of (at least) $q^k$ possible contents of row~$i$ in~$A$
must form a subset of~$\Sigma_q^n$
that, by itself (and independently of the contents of the other rows
in~$A$), meets the correction--detection capabilities.

Secondly, note that a given computed $k$-prefix $\bldc' = \bldu A'$
can be associated with different $(n{-}k)$-suffixes
(redundancy symbols) $\bldc'' = \bldu A''$, depending on~$\bldu$.
This is different from
the common coding theory setting, where the redundancy symbols
are uniquely determined by the information symbols\footnote{%
Moreover, while systematic encoding is a matter of preference
in ordinary coding applications, in our setting it is actually
a necessity:
the benefits of using the DPE would diminish if post-processing
of its output were required even when the output were error-free.}
(the latter being the counterparts of the entries of~$\bldc'$
in our setting).
For instance, if~$A$ in Example~\ref{ex:paritycode} is
\[
A =
\left(
\begin{array}{ccc|c}
1 & 0 & 0 & 1 \\
0 & 1 & 0 & 1 \\
1 & 1 & 0 & 0
\end{array}
\right)
\]
(where $q = 2$, $k = 3$, and $n = 4$), then,
for $\bldu = (0 \; 0 \; 1)$,
\[
(0 \; 0 \; 1) \, A =  (1 \; 1 \; 0 \; 0)
\]
while for $\bldu = (1 \; 1 \; 0)$,
\[
(1 \; 1 \; 0) \, A =  (1 \; 1 \; 0 \; 2)
\]
(in both cases, $\bldc' = (1 \; 1 \; 0)$).
Indeed, we will see in the sequel some coding schemes
where we will be able to recover~$\bldc'$ correctly out of~$\bldy$
(which will suffice for our purposes), yet we will not necessarily
recover~$\bldc''$. This means that we will need to present
the error correction--detection specification of a DPE coding scheme
slightly differently than usual;
we do this in Section~\ref{sec:definitions} below.

In Section~\ref{sec:oneerror}, we present methods for
single-error correction and double-error detection
in the $L_1$-metric. Methods for multiple-error correction
for that metric are then discussed in Section~\ref{sec:multipleerrors}.
Finally, the Hamming metric is considered in Section~\ref{sec:Hamming}.
We will mainly focus on a regime where the number~$\tau$
of correctable errors is fixed (i.e., small) while~$n$ grows.
Under these conditions, the required redundancy, $n-k$,
of our methods will be of
the order of $\tau \cdot \log_q n$ in the case of the $L_1$-metric,
and approximately twice that number in the case of the Hamming metric.
Moreover, both the encoding and decoding can be efficiently implemented;
in particular, the decoding requires a number of
integer (or finite field) arithmetic operations which is proportional
to $\tau n$ (and the implementation can be parallelized to
a latency proportional to~$\tau$),
where the operands are of the order of $\log_2 n$ bits long.

\section{Definitions}
\label{sec:definitions}

For integer vectors~$\bldx_1$ and~$\bldx_2$ of the same length,
we denote by $\distance_\Lee(\bldx_1,\bldx_2)$ the $L_1$-distance
between them, namely,
$\distance_\Lee(\bldx_1,\bldx_2) = \| \bldx_1 - \bldx_2 \|$.
The \emph{Manhattan sphere} of radius~$t$ centered at
$\bldy \in \Integers^n$ is defined as the set of all vectors
in~$\Integers^n$ at $L_1$-distance at most~$t$ from~$\bldy$:
\[
\Sphere_\Lee(\bldy,t) =
\left\{
\bldx  \in \Integers^n \;:\; \distance_\Lee(\bldx,\bldy) \le t
\right\} \; .
\]
The \emph{volume} (size) of $\Sphere_\Lee(\bldy,t)$ is known
to be~\cite{GW1}, \cite{GW2}:
\begin{equation}
\label{eq:volumeLee}
\Volume_\Lee(n,t) = \sum_{i=0}^{\min \{ t, n \}}
2^i {n \choose i} {t \choose i}
\; .
\end{equation}
In particular, $\Volume_\Lee(n,1) = 2n+1$, and for any fixed~$t$
and sufficiently large~$n$ we have
$\Volume_\Lee(n,t) = O(n^t)$, where the hidden constant depends on~$t$.

Turning to the Hamming metric,
we denote by $\distance_\Hamming(\bldx_1,\bldx_2)$
the Hamming distance between~$\bldx_1$ and~$\bldx_2$,
and the Hamming sphere of radius~$t$ centered at
$\bldy \in \Integers^n$ is defined by
\[
\Sphere_\Hamming(\bldy,t) =
\left\{
\bldx  \in \Integers^n \;:\; \distance_\Hamming(\bldx,\bldy) \le t
\right\}
\]
(which has infinite size when $t > 0$).
In what follows, we will sometimes omit the identifier
``$\Lee$'' or ``$\Hamming$'' from $\distance(\cdot,\cdot)$
and $\Sphere(\cdot,\cdot)$, if the text applies to both metrics.

Given $\Sigma_q$ and positive integers~$\ell$, $n$, and $k < n$,
a DPE \emph{coding scheme} is a pair $(\encoder,\decoder)$, where
\begin{itemize}
\item
$\encoder : \Sigma_q^{\ell \times k} \rightarrow
\Sigma_q^{\ell \times n}$
is an \emph{encoding mapping}
such that for every $A' \in \Sigma_q^{\ell \times k}$,
the image $A = \encoder(A')$ has the form
$( A' \;|\; A'' )$  for some $A'' \in \Sigma_q^{\ell \times (n-k)}$.
The set
\[
\code =
\left\{
\bldu \, \encoder(A') \;:\;
A' \in \Sigma_q^{\ell \times k},  \; \bldu \in \Sigma_q^\ell
\right\}
\]
is the \emph{code induced by~$\encoder$} and its members
are called \emph{codewords}.
Thus, $\code \subseteq \Sigma_Q^n$, where $Q = \ell(q{-}1)^2 + 1$.
\item
$\decoder : \Sigma_Q^n \rightarrow \Sigma_Q^k \cup \{ \failure \}$ is
a \emph{decoding mapping}
(the return value~$\failure$ will designate a decoding failure).
\end{itemize}

Note that in the above definition,
the decoding mapping~$\decoder$ is not
a function of~$A'$ (yet one could consider also a different setting
where~$A'$ is known to the decoder).

Borrowing (somewhat loosely) classical coding terms, we will refer
to~$n$ and~$k$ as the \emph{length} and \emph{dimension}, respectively,
of the coding scheme. In the context of a given coding scheme,
the $k$-prefix (respectively, $(n{-}k)$-suffix)
of a vector $\bldx \in \Integers^n$ will be denoted hereafter by
$\bldx'$ (respectively, $\bldx''$).
This notational convention extends
to $\ell \times n$ matrices over~$\Integers$,
with~$A'$ (respectively, $A''$) standing for the sub-matrix
consisting of the first $k$~columns
(respectively, last $n-k$ columns)
of an $\ell \times n$ matrix~$A$ over~$\Integers$.
Denoting row~$i$ of a matrix~$X$ by~$X_i$,
we then have $(A')_i = (A_i)'$ and $(A'')_i = (A_i)''$, for
every $i \in \Int{\ell}$.

Given nonnegative integers~$\tau$ and~$\sigma$,
a coding scheme $(\encoder,\decoder)$ is said to
\emph{correct~$\tau$ errors and detect $\tau + \sigma$ errors}
(in the $L_1$-metric or the Hamming metric, depending
on the context) if the following conditions hold for every
computed vector $\bldc = \bldu A \in \code$
and the respective read vector\footnote{%
It is assumed hereafter that the entries of the received vector remain
in the same alphabet, $\Sigma_Q$, as of the computed vector;
while errors could push the entries to outside that range,
they can always be coerced back into~$\Sigma_Q$.}
$\bldy \in \Sigma_Q^n$.
\begin{itemize}
\item
\emph{(Correction condition)}
If $\distance(\bldy,\bldc) \le \tau$, then
$\decoder(\bldy) = \bldc'$.
\item
\emph{(Detection condition)}
Otherwise, if $\distance(\bldy,\bldc) \le \tau + \sigma$, then
$\decoder(\bldy) \in \{ \bldc', \failure \}$.
\end{itemize}
That is, if the number of errors is~$\tau$ or less,
then the decoder must produce the correct result of
the target computation; otherwise, if the number of errors is
$\tau + \sigma$ or less, the decoder can flag decoding failure
instead (but it cannot produce an incorrect result).

So, unlike the respective conditions for ordinary codes,
the sphere $\Sphere(\bldy,\tau)$ may contain multiple codewords
of~$\code$, yet they all must agree on their $k$-prefixes.
Similarly, the sets $\Sphere(\bldy,\tau)$
and $\Sphere(\bldy,\tau{+}\sigma) \setminus \Sphere(\bldy,\tau)$
may both contain codewords of~$\code$, yet these codewords
must agree on their $k$-prefixes.

By properly defining the \emph{minimum distance} of~$\code$,
we can extend to our setting the well known relationship
between minimum distance and correction capability.
Specifically,
the minimum distance of~$\code$, denoted $\distance(\code)$
(with an identifier of the particular metric used),
is defined as the smallest distance between any two codewords
in~$\code$ having distinct $k$-prefixes:
\[
\distance(\code) =
\min_{\bldc_1, \bldc_2 \in \code : \atop \bldc_1' \ne \bldc_2'}
\distance(\bldc_1, \bldc_2) \; .
\]
The following result then extends from the ordinary coding
setting~\cite[p.~14, Prop.~1.5]{Roth}
(for completeness, we include a proof
in Appendix~\ref{sec:proofs}).

\begin{proposition}
\label{prop:distance}
Let $\encoder : \Sigma_q^{\ell \times k} \rightarrow
\Sigma_q^{\ell \times n}$ be an encoding mapping
with an induced code~$\code$, and let~$\tau$ and~$\sigma$
be nonnegative integers such that
\[
2 \tau + \sigma < \distance(\code) \; .
\]
Then there exists a decoding mapping
$\decoder : \Sigma_Q^n \rightarrow \Sigma_Q^k \cup \{ \failure \}$
such that the coding scheme $(\encoder,\decoder)$
can correct~$\tau$ errors and detect $\tau + \sigma$ errors.
\end{proposition}

For the special case of the Hamming metric,
Proposition~\ref{prop:distance} can be generalized to handle
erasures as well
(see~\cite[p.~16, Prop.~1.7]{Roth} and Appendix~\ref{sec:proofs}).

\begin{proposition}
\label{prop:distanceHamming}
With~$\encoder$ and~$\code$ as in Proposition~\ref{prop:distance},
let~$\tau$, $\sigma$, and~$\rho$
be nonnegative integers such that
\[
2 \tau + \sigma + \rho < \distance_\Hamming(\code) \; .
\]
Then there exists a decoding mapping
$\decoder : \Sigma_Q^n \rightarrow \Sigma_Q^k \cup \{ \failure \}$
such that the coding scheme $(\encoder,\decoder)$
can correct~$\tau$ errors and detect $\tau + \sigma$ errors,
in the presence of up to~$\rho$ erasures.
\end{proposition}

The coding schemes that we present in upcoming sections are based
on known codes, in particular known schemes for
the Lee and Manhattan metrics---primarily
Berlekamp codes~\cite[Ch.~9]{Berlekamp}, \cite[Ch.~10]{Roth}.
Yet certain adaptations are needed due to the fact that the computation
of the redundancy symbols of the codewords in the induced code
$\code = \{ \bldc = \bldu \, \encoder(A') \}$ has to be done
only through the computation of $A' \mapsto \encoder(A')$
(which is independent of~$\bldu$). Moreover,
the alphabet, $\Sigma_q$, of the entries of $\encoder(A')$ is smaller
than the alphabet, $\Sigma_Q$, of the codewords in~$\code$.
Our coding schemes will be \emph{separable}, in the sense
that for each row index $i \in \Int{\ell}$,
the contents $(\encoder(A'))_i$ will only be a function of $A_i'$
(and not of the rest of the rows in $A'$); in fact, that function
will be the same for all~$i$, and will not depend on~$\ell$.
It is expected, however, that the designed number of correctable errors,
$\tau$, will tend to increase with~$\ell$.

\section{Single error correction in the $L_1$-metric}
\label{sec:oneerror}

In this section, we describe a DPE coding scheme,
$(\encoder_1,\decoder_1)$, for correcting one $L_1$-metric error;
this scheme will then be extended
(in Subsection~\ref{sec:doubleerrordetection})
to also detect two errors.

\subsection{The coding scheme}
\label{sec:singleerror}

Given an alphabet size $q \ge 2$ and a code length~$n$,
we let $m = \lceil \log_q (2n+1) \rceil$ and $k = n-m$
(thus, $m$ will be the redundancy).
Let
\[
\bldalpha =
\left( \alpha_0 \; \alpha_1 \; \ldots \; \alpha_{n-1} \right)
\]
be a vector in~$\Integers^n$ that satisfies the following properties.
\begin{list}{}{\settowidth{\labelwidth}{\textit{iii}}}
\item[(i)]
The entries of~$\bldalpha$ are nonzero distinct
elements in $\Int{2n{+}1}$.
\item[(ii)]
For any two indexes $i, j \in \Int{n}$,
\[
\alpha_i + \alpha_j \ne 2n+1 \; .
\]
\item[(iii)]
$\alpha_{k+j} = q^j$, for $j \in \Int{m}$.
\end{list}

We will refer to the entries of~$\bldalpha$ as \emph{code locators}.
Code locators that satisfy conditions (i)--(iii)
can be easily constructed for every $q \ge 2$: e.g.,
when $q^{m-1} \le n$, we can take
\[
\left\{ \alpha_j \right\}_{j \in \Int{n}} =
\{ 1, 2, 3, \ldots, n \} \; ,
\]
otherwise,
\[
\left\{ \alpha_j \right\}_{j \in \Int{n}} =
(\{ 1, 2, 3, \ldots, n \} \setminus \{ 2n{+}1{-}q^{m-1} \})
\cup \{ q^{m-1} \}
\]
(note that $q^{m-1} < 2n+1$ and that $q^{m-2} \le n$,
yet $q^{m-1}$ may be larger than~$n$; in fact, this will always be
the case when $q = 2$).

The encoding mapping
$\encoder_1 :
\Sigma_q^{\ell \times k} \rightarrow \Sigma_q^{\ell \times n}$
is defined as follows: for every
$A' = (a_{i,j})_{i \in \Int{\ell}, j \in \Int{k}}$,
the last $m$~columns in
$A = ( A' \;|\; A'' ) = \encoder_1(A')$ are set so that
\begin{eqnarray}
\lefteqn{
\sum_{j \in \Int{m}} a_{i,k+j}
\underbrace{\alpha_{k+j}}_{q^j}
} \makebox[2ex]{} \nonumber \\
\label{eq:oneerrorencoder}
& = &
\Bigl( -\sum_{j \in \Int{k}} a_{i,j} \alpha_j \Bigr)
\; \Mod \; (2n+1) \; ,
\quad i \in \Int{\ell} \; ,
\end{eqnarray}
where the remainder (the result of the ``$\Mod$'' operation) is taken
to be in $\Int{2n{+}1}$. Simply put,
$a_{i,n-1} a_{i,n-2} \ldots a_{i,k+1} a_{i,k}$ is
the representation to base~$q$ (from the most-significant digit
down to the least) of the remainder in $\Int{2n{+}1}$
of the (nonpositive) integer $-\sum_{j \in \Int{k}} a_{i,j} \alpha_j$,
when divided by~$2n+1$.

It follows from~(\ref{eq:oneerrorencoder}) that each
$A = \encoder_1(A')$ satisfies
\[
A \bldalpha^\transpose \equiv \bldzero \quad (\mod \; (2n+1)) \; ,
\]
where~$(\cdot)^\transpose$ denotes transposition and
the congruence holds component-wise. Hence, for each
codeword $\bldc = \bldu A$ in the induced code~$\code$ we have
\begin{equation}
\label{eq:oneerror}
\bldc \cdot \bldalpha^\transpose \equiv
\bldu A \bldalpha^\transpose \equiv 0 \quad (\mod \; (2n+1)) \; .
\end{equation}
This, in turn, implies that for every two distinct codewords
$\bldc_1, \bldc_2 \in \code$,
\[
(\bldc_1 - \bldc_2) \cdot \bldalpha^\transpose \equiv 0
\quad (\mod \; (2n+1)) \; ,
\]
and, therefore, by conditions~(i)--(ii) we get that
$\distance_\Lee(\bldc_1,\bldc_2) = \| \bldc_1 - \bldc_2 \| > 2$, namely,
that $\distance_\Lee(\code) \ge 3$. We conclude from
Proposition~\ref{prop:distance} that when using
the encoding mapping defined by~(\ref{eq:oneerrorencoder}),
to map $A'$ into $A = \encoder_1(A')$,
we should be able to correct one error; alternatively, we should be able
to detect two errors. We demonstrate next
a single-error-correcting decoding mapping.

Let $\bldy = (y_j)_{j \in \Int{n}} = \bldc + \blde = \bldu A + \blde$ be
the read vector at the output of the DPE, where
$\blde \in \Integers^n$ is an error vector having at most
one nonzero entry, equaling~$\pm 1$.
The decoding will start by computing
the syndrome of~$\bldy$, which is defined by
\begin{eqnarray*}
s
& = &
\left( \bldy \cdot \bldalpha^\transpose \right) \; \Mod \; (2n+1) \\
& = &
\Bigl( \sum_{j \in \Int{n}} y_j \alpha_j \Bigr)
\; \Mod \; (2n+1) \; .
\end{eqnarray*}
We then have,
\[
s \equiv
\bldu A \bldalpha^\transpose + \blde \cdot \bldalpha^\transpose
\stackrel{(\ref{eq:oneerror})}{\equiv}
\blde \cdot \bldalpha^\transpose \quad (\mod \; (2n+1)) \; .
\]
It follows that $s = 0$ when $\blde = \bldzero$;
otherwise, if~$\blde$ has~$\pm 1$ at position~$j$
(and is zero otherwise), then
\[
s \equiv \pm \alpha_j \quad (\mod \; (2n+1)) \; .
\]
Hence, due to conditions~(i)--(ii), the syndrome~$s$ identifies
the error location~$j$ and the error sign uniquely.

The encoding and decoding procedures for
our single-error correction scheme
are summarized in Figures~\ref{fig:oneerrorencoder}
and~\ref{fig:oneerrordecoder}.

\begin{figure}[hbt]
\centering
\begin{myalgorithm}
\vspace{-1ex}

\quad
\textbf{Input:}
$\ell \times k$ matrix
$A' = (a_{i,j})_{i \in \Int{\ell}, j \in \Int{k}}$ over~$\Sigma_q$.

\quad
\textbf{Output:}
$\ell \times n$ matrix
$(A' \;|\; A'') = (a_{i,j})_{i \in \Int{\ell}, j \in \Int{n}}$
over~$\Sigma_q$.

\quad
$//$
$m = \lceil \log_q (2n+1) \rceil$, $k = n-m$.

\quad
$//$
$\bldalpha$ satisfies conditions (i)--(iii).

\begin{clause}
\item
For all $i \in \Int{\ell}$ do $\{$
\begin{clause}
\item
Set $(a_{i,n-1} \; a_{i,n-2} \; \ldots \; a_{i,k+1} \; a_{i,k})$
to satisfy Eq.~(\ref{eq:oneerrorencoder}).
\end{clause}
$\}$
\end{clause}

\end{myalgorithm}
\caption{Encoding mapping
$\encoder_1 : A' \mapsto (A' \;|\; A'')$
for single-error correction (or double-error detection).}
\label{fig:oneerrorencoder}
\end{figure}

\begin{figure}[hbt]
\centering
\begin{myalgorithm}
\vspace{-1ex}

\quad \textbf{Input:}
$\bldy = (\bldy' \;|\; \bldy'') \in \Sigma_q^n$.

\quad \textbf{Output:}
$\bldw = (w_j)_{j \in \Int{k}} \in \Sigma_q^k$,
or~$\failure$ (decoding failure).

\quad $//$
Parameters are as defined in Figure~\protect\ref{fig:oneerrorencoder}.

\begin{clause}
\item
Let $\bldw \leftarrow \bldy'$;

\item
Let
\[
s \leftarrow
\left( \bldy \cdot \bldalpha^\transpose \right) \; \Mod \; (2n+1) \; ;
\]

\item
If $s = 0$ then $\{$
\begin{clause}
\item
;
\quad $//$
$\bldy$ is error-free
\end{clause}
$\}$

\item
Else if $s = \alpha_j$ for some~$j \in \Int{n}$, then $\{$
\begin{clause}
\item
If $j \in \Int{k}$ then let $w_j \leftarrow w_j + 1$;
\end{clause}
$\}$
\item
Else if $s = 2n{+}1{-}\alpha_j$
for some~$j \in \Int{n}$, then $\{$
\begin{clause}
\item
If $j \in \Int{k}$ then let $w_j \leftarrow w_j - 1$;
\end{clause}
$\}$
\item
Else $\{$
\begin{clause}
\item
Return $\failure$.
\end{clause}
$\}$
\end{clause}

\end{myalgorithm}
\caption{Decoding mapping $\decoder_1 : \bldy \mapsto \bldw$
for single-error correction.}
\label{fig:oneerrordecoder}
\end{figure}

\begin{example}
\label{ex:oneerrorencoder}
Let $q = 2$ and $n = 15$, in which case $m = 5$ and $k = 10$.
Select $\ell = 3$ and
\[
\bldalpha =
\left( \, 3 \; 5 \; 6 \; 7 \; 9 \; 10 \; 11 \; 12 \; 13 \; 14 \bigm|
1 \; 2 \; 4 \; 8 \; 16 \,
\right)
\]
(which satisfies conditions~(i)--(iii)). Suppose that
$A'$ is the following $3 \times 10$ matrix:
\[
A' =
\left(
\arraycolsep0.2ex
\begin{array}{cccccccccc}
  1 & 0 & 1 & 1 & 0 & 1 & 0 & 0 & 1 & 0 \\
  0 & 0 & 0 & 1 & 0 & 1 & 1 & 0 & 0 & 1 \\
  0 & 1 & 0 & 0 & 0 & 1 & 0 & 1 & 1 & 1
\end{array}
\right)
\; .
\]
For $i = 0, 1, 2$, the values at
the right-hand side of~(\ref{eq:oneerrorencoder})
are given by
\begin{eqnarray*}
-(3 + 6 + 7 + 10 + 13) \; \Mod \; 31 & = & 23 \\
-(7 + 10 + 11 + 14) \; \Mod \; 31 & = & 20 \\
-(5 + 10 + 12 + 13 + 14) \; \Mod \; 31 & = & 8 \; ,
\end{eqnarray*}
and, so,
\[
A = \encoder_1(A') =
\left(
\arraycolsep0.2ex
\begin{array}{cccccccccc@{\;}|@{\;}ccccc}
  1 & 0 & 1 & 1 & 0 & 1 & 0 & 0 & 1 & 0 & 1 & 1 & 1 & 0 & 1 \\
  0 & 0 & 0 & 1 & 0 & 1 & 1 & 0 & 0 & 1 & 0 & 0 & 1 & 0 & 1 \\
  0 & 1 & 0 & 0 & 0 & 1 & 0 & 1 & 1 & 1 & 0 & 0 & 0 & 1 & 0
\end{array}
\right)
\; .
\]
For $\bldu = ( 1 \; 1 \; 1 )$, we get
\[
\bldc = \bldu A =
\left( \, 1\; 1\; 1\; 2\; 0\; 3\; 1\; 1\; 2\; 2 \bigm| 1\; 1\; 2\; 1\; 2
\, \right) \; .
\]
Suppose that the read vector is
\[
\bldy =
\left( \,
1\; 1\; 1\; 2\; 0\; 2\; 1\; 1\; 2\; 2 \bigm| 1\; 1\; 2\; 1\; 2
\, \right) \; .
\]
The syndrome of~$\bldy$ is given by
\begin{eqnarray*}
s & = &
\left( \bldy \cdot \bldalpha^\transpose \right) \; \Mod \; 31 \\
& = &
(3 {+} 5{+}6{+}2{\cdot}7{+}0{\cdot}9{+}2{\cdot}10{+}11{+}12 \\
&&
\quad
{+}2{\cdot}13
{+}2{\cdot}14{+}1{+}2{+}2{\cdot}4{+}8{+}2{\cdot}16) \; \Mod \; 31 \\
& = &
21 \; .
\end{eqnarray*}
Namely, $s = 31-21 = 10 = \alpha_5$, indicating that the error
location is $j = 5$ (the sixth entry) and the error value is~$-1$
(corresponding to changing the value~$3$ into~$2$).\qed
\end{example}

We end this subsection by demonstrating that
a redundancy of $n-k = \lceil \log_q (2n+1) \rceil$ is
within one symbol from the smallest possible
for any coding scheme that corrects one error in the $L_1$-metric.
Recall that by taking~$\bldu$ as a unit vector
it follows that for any row index~$i$,
the set of the~$q^k$ possible contents of~$A_i$
forms an (ordinary) code $\code_i \subseteq \Sigma_q^n$
that is capable of correcting one error.
Hence, by a sphere-packing argument we conclude that
for distinct~$\bldc \in \code_i$,
the (truncated) spheres $\Sphere_\Lee(\bldc,1) \cap \Sigma_q^n$
must be disjoint subsets of~$\Sigma_q^n$.
Yet $|\Sphere_\Lee(\bldc,1) \cap \Sigma_q^n| \ge n+1$,
and, so, $q^k = |\code_i| \le q^n/(n+1)$, namely,
we have the lower bound
\[
n-k \ge \lceil \log_q (n+1) \rceil \; .
\]

\subsection{Allowing additional error detection}
\label{sec:doubleerrordetection}

The presented coding scheme can be easily enhanced so
that the induced code has minimum distance~$4$; namely,
the scheme can detect two errors on top of correcting a single error,
or, alternatively, it can detect three errors with no attempt to correct
any error. We do this by extending the code length by~$1$ and
adding a parity bit to each row of $A = \encoder_1(A')$, as we did
in Example~\ref{ex:paritycode} (with~$A$ playing the role of~$A'$
therein). This, in turn, allows the decoder to recover the parity
of the number of errors (whether it was even or odd).
An odd number is seen as one error, and the algorithm
in Figure~\ref{fig:oneerrordecoder} is then applied.
Otherwise, a zero syndrome will indicate no errors, while
a nonzero syndrome indicates two errors (which will be flagged
by~$\failure$).

When $q > 2$, this extra error detection capability can sometimes be
achieved \emph{without} increasing the redundancy. To see this,
consider first the case where~$q$ is odd:
redefine~$m$ to be $\lceil \log_q (4n+2)) \rceil$
(depending on~$n$, the value of~$m$ may remain unchanged by
this redefinition), and modify condition~(i)--(ii) as follows.
\begin{list}{}{\settowidth{\labelwidth}{\textit{iii'}}}
\item[(i')]
The entries of~$\bldalpha$ are \emph{odd} distinct
elements in $\Int{4n{+}2}$.
\item[(ii')]
For any two indexes $i, j \in \Int{n}$,
\[
\alpha_i + \alpha_j \ne 4n+2 \; .
\]
\end{list}
(Note that condition~(iii), which remains unchanged, is consistent
with condition~(i'). Also, condition~(ii') disqualifies
$2n+1$ to be an entry\footnote{%
\label{fn:violationqodd}
Yet the coding scheme will work also when
$(\alpha_{n-1} =) \; q^{m-1} = 2n+1$,
in spite of violating condition~(ii').}
of~$\bldalpha$.) The encoding is
similar to~(\ref{eq:oneerrorencoder}), except that
the remainder at the right-hand side is now computed modulo~$4n+2$.
Accordingly, during decoding, the syndrome is redefined to
\[
s \leftarrow
\left( \bldy \cdot \bldalpha^\transpose \right) \; \Mod \; (4n+2) \; ,
\]
and, so, the parity of the syndrome equals the parity of
the number of errors. An odd syndrome indicates that one error
has occurred, in which case the error location and sign can be
recovered from the value of~$s$. A nonzero even syndrome~$s$
indicates that two errors have occurred.

Assume now that~$q$ is an even integer greater than~$2$.
In this case, condition~(i') would contradict condition~(iii),
as the latter requires that the last $m{-}1$ entries of
$\bldalpha$ be even. To overcome this impediment, we will
modify the definition of~$m$ and rewrite condition~(iii).
Specifically, we let~$m$ be the smallest integer
that satisfies $f_m(q) \ge 4n+2+(-1)^m$ where, for every
nonnegative $j \in \Integers$,
\begin{equation}
\label{eq:fjq}
f_j(q) = \frac{q^{j+1} + (-1)^j}{q+1} \; .
\end{equation}
Note that $f_0(q) = 1$ and that for every $j > 0$,
\[
f_j(q) = (q-1) \sum_{i \in \Int{j}} f_i(q) +
\left\{
\begin{array}{lcl}
1 && \textrm{if~$j$ is even} \\
0 && \textrm{otherwise}
\end{array}
\right.
\; ,
\]
which means that every integer in $\Int{4n{+}2}$
has a representation of the form\footnote{%
The sequence $(f_j(q))_j$ can be seen as a generalization
of the Jacobsthal sequence: see for instance~\cite{Hoaradam}.}
$\sum_{j \in \Int{m}} b_j f_j(q)$, for $b_j \in \Sigma_q$.
Moreover, the values $f_j(q)$ are all odd.
Hence, rewriting condition~(iii) as follows will be
consistent\footnote{%
\label{fn:violationqeven}
The coding scheme will work also when
$(\alpha_{n-1} =) \; f_{m-1}(q) = 2n+1$,
in spite of violating condition~(ii').}
with condition~(i'):
\begin{list}{}{\settowidth{\labelwidth}{\textit{iii'}}}
\item[(iii')]
$\alpha_{k+j} = f_j(q)$, for $j \in \Int{m}$.
\end{list}
   From this point onward, we proceed as in the case of odd~$q$.
(We point out that since $f_m(q) < q^m$ for $m > 0$,
the inequality $f_m(q) \ge 4n+2+(-1)^m$ is generally stronger than
$m \ge \log_q (4n+2)$; however, the ratio $f_m(q)/q^m$
does approach~$1$ as $q \rightarrow \infty$.)

\begin{example}
\label{ex:evenq>2}
Suppose that $q = 8$ and $n = 13$.
Since $f_1(8) = 7$ and $f_2(8) = 57$, we can take $m = 2$ and
\[
\bldalpha =
\left( \, 3 \; 5 \; 9 \; 11 \; 13 \; 15 \; 17 \; 19 \; 21 \; 23 \; 25
\bigm|
1 \; 7 \, \right)
\; ,
\]
resulting in a single-error-correcting double-error-detecting
coding scheme. The redundancy $n-k$ will be only~$2$ in this case
(as opposed to~$3$ had we added a parity bit
to the construction in Subsection~\ref{sec:singleerror}
for $n = 13$).\qed
\end{example}

\section{Larger minimum $L_1$-distances}
\label{sec:multipleerrors}

In this section, we show how to extend
the construction of Section~\ref{sec:oneerror} to
correct more errors in the $L_1$-metric. Our coding schemes
will make use of known construction for the Lee metric,
specifically Berlekamp codes,
to be recalled in the next subsection.

\subsection{Lee-metric codes}
\label{sec:Lee}

Let~$p$ be an odd prime and let $F = \GF(p)$.
Representing the elements of~$F$ as
$0, 1, 2, \ldots, p{-}1$, the last $(p{-}1)/2$ elements
in this list will be referred to as the ``negative'' elements in~$F$.
The \emph{Lee metric} over~$F$ is defined similarly to
the $L_1$-metric over~$\Integers$, using the following
definition of the absolute value (in~$\Integers$) of
an element $z \in F$:
\[
|z| =
\left\{
\begin{array}{ccl}
z   && \textrm{if~$z$ is ``nonnegative''} \\
p-z && \textrm{otherwise}
\end{array}
\right.
\; .
\]

Let~$n$ and~$\tau$ be positive integers such that $2 \tau < p$,
and let $h = \lceil \log_p (2n+1) \rceil$.
Also, let
\[
\bldbeta = ( \beta_0 \; \beta_1 \; \ldots \; \beta_{n-1} )
\]
be a vector of length~$n$
over the extension field $\Phi = \GF(p^h)$ whose entries are
nonzero and distinct and satisfy
$\beta_i + \beta_j \ne 0$ for every $i, j \in \Int{n}$
(compare with conditions~(i)--(ii) in Section~\ref{sec:oneerror};
it is easy to see that here, too, such
a vector~$\bldbeta$ always exists). The respective Berlekamp code,
$\Code_\Ber = \Code_\Ber(\bldbeta,\tau)$, is
defined as the set of all row vectors in~$F^n$ in the right kernel
of the following $\tau \times n$ parity-check matrix,
$H_\Ber = H_\Ber(\bldbeta,\tau)$, over~$\Phi$:
\begin{equation}
\label{eq:HBer}
H_\Ber =
\left(
\renewcommand{\arraystretch}{1.1}
\begin{array}{cccc}
\beta_1           & \beta_2           & \ldots & \beta_n          \\
\beta_1^3         & \beta_2^3         & \ldots & \beta_n^3        \\
\beta_1^5         & \beta_2^5         & \ldots & \beta_n^5        \\
\vdots            & \vdots            & \vdots & \vdots           \\
\beta_1^{2\tau-1} & \beta_2^{2\tau-1} & \ldots & \beta_n^{2\tau-1}
\end{array}
\right) \; .
\end{equation}
That is,
\[
\Code_\Ber =
\left\{
\bldc \in F^n \;:\; \bldc \cdot H_\Ber^\transpose = \bldzero
\right\}
\; .
\]
Thus, $\Code_\Ber$ is a linear $[n,k]$ code over~$F$
with a redundancy $n-k$
of at most~$\tau h = \tau \lceil \log_p (2n+1) \rceil$.

The minimum Lee distance of~$\Code_\Ber$ is known to be
at least $2\tau+1$, and there are known efficient algorithms
for decoding up to~$\tau$ Lee-metric errors
(see~\cite[Ch.~9]{Berlekamp} and~\cite[\S 10.6]{Roth}).
These decoders typically start with computing the syndrome,
$\blds = \bldy H_\Ber^\transpose$, of
the received vector $\bldy \in F^n$, and then implement
a function $\Decoder_\Ber : \Phi^\tau \rightarrow F^n$
which maps~$\blds$ to the error vector $\blde = \Decoder_\Ber(\blds)$.
The condition $2 \tau < p$ implies that~(\ref{eq:volumeLee}) is also
the volume of a Lee-metric sphere of radius~$t$ in~$F^n$.
Hence, by sphere-packing arguments,
the size of any Lee-metric $\tau$-error-correcting code in~$F^n$
is bounded from above by $p^n/\Volume(n,\tau)$~\cite[p.~318]{Roth}.
Thus, up to an additive term that depends on~$\tau$
(but not on~$n$) the dimension of~$\Code_\Ber$ is the largest possible,
for a given length~$n$ and number~$\tau$ of Lee-metric errors
to be corrected.

There is a close relationship between the construction presented
in Section~\ref{sec:oneerror} and Berlekamp codes. Specifically,
when~$n$ is taken so that $p = 2n+1$ is a prime,
then each row of $A = \encoder_1(A')$
is a codeword of $\Code_\Ber(\bldalpha,1)$,
assuming that the entries of~$A$ and~$\bldalpha$ are seen as elements of
$\Phi = F = \GF(2n{+}1)$. Consequently, the induced code~$\code$
forms a subset of $\Code_\Ber(\bldalpha,1)$.

With this relationship in mind, we will next present
coding schemes whose induced codes have
minimum $L_1$-distances~$5$ and above.

\subsection{Double-error-correcting coding scheme}
\label{sec:twoerrors}

In this subsection, we present a DPE coding scheme for correcting
two errors in the $L_1$-metric (namely, the induced code
will have minimum $L_1$-distance at least~$5$). This scheme will
then be extended to also detect three errors (minimum distance~$6$).

Given the alphabet~$\Sigma_q$ and the number of rows~$\ell$,
let $p > 3$ be a prime, and define
$n_1 = (p-1)/2$,
$m = \lceil \log_q p \rceil$,
$n_2 = n_1 + m$, and $n = n_2 + 1$.
The coding scheme will have
dimension $k = n_1-m$, length\footnote{%
The seeming restriction on~$n$ imposed by requiring that $2n_1+1$
is a prime can be lifted by code shortening.}
$n$ and, therefore, redundancy $n - k = 2m + 1$. The encoding mapping,
$\encoder_2 :
\Sigma_q^{k \times \ell} \rightarrow \Sigma_Q^{n \times \ell}$,
will take the form of a composition
\[
\encoder_2 = \hat{\varphi}_2 \circ \varphi_2 \circ \encoder_1 \; ,
\]
where the component functions are defined
in Figure~\ref{fig:twoerrorencoder}.
\begin{figure}[hbt]
\centering
\begin{myalgorithm}
\vspace{-1ex}
\begin{itemize}
\item
$\encoder_1$ is the encoding mapping in
Figure~\ref{fig:oneerrorencoder}, with~$n$ therein replaced by~$n_1$
(in particular, the remainder in the right-hand side
of~(\ref{eq:oneerrorencoder}) is computed modulo~$p$).
\item
$\varphi_2 :
\Sigma_q^{\ell \times n_1} \rightarrow \Sigma_q^{\ell \times n_2}$
maps an $\ell \times n_1$ matrix
$A' = (a_{i,j})_{i \in \Int{\ell}, j \in \Int{n_1}}$
over~$\Sigma_q$ to $A = ( A' \;|\; A'' )$,
where the last~$m$ columns in~$A$ are set so that
\begin{equation}
\label{eq:varphi2}
\sum_{j \in \Int{m}} a_{i,n_2+j} q^j =
\Bigl( \sum_{j \in \Int{n_1}} a_{i,j} \alpha_j^3 \Bigr)
\; \Mod \; p \; ,
\quad i \in \Int{\ell} \; .
\end{equation}
\item
$\hat{\varphi}_2 :
\Sigma_q^{\ell \times n_2} \rightarrow \Sigma_q^{\ell \times n}$
is the parity mapping defined on the last~$m$ columns
of the argument matrix; that is,
an $\ell \times n_2$ matrix
$A' = (a_{i,j})_{i \in \Int{\ell}, j \in \Int{n_2}}$
over~$\Sigma_q$ is mapped to $A = ( A' \;|\; A'' )$,
where the entries of the last column in~$A$ are given by
\[
a_{i,n_2}
= \Bigl( \sum_{j \in \Int{m}} a_{i,n_1+j} \Bigr) \; \Mod \; 2 \; ,
\quad i \in \Int{\ell} \; .
\]
\end{itemize}
\end{myalgorithm}
\caption{Component functions of a double-error-correcting encoding
mapping
$\encoder_2 = \hat{\varphi}_2 \circ \varphi_2 \circ \encoder_1$.}
\label{fig:twoerrorencoder}
\end{figure}

\begin{example}
\label{ex:twoerrorencoder}
Let $q = 2$ and $p = 31$, in which case $n_1 = 15$, $m = 5$,
$n_2 = 20$, $n = 21$, and $k = 10$.
Select $\ell = 3$ and let~$\bldalpha$
and~$A'$ be as in Example~\ref{ex:oneerrorencoder}.
For $i = 0, 1, 2$,
the values at the right-hand side of~(\ref{eq:varphi2}) are
\begin{eqnarray*}
(3^3{+}6^3{+}7^3{+}10^3{+}13^3{+}1^3{+}2^3{+}4^3{+}16^3) \; \Mod \; 31
& = & 16 \\
(7^3 + 10^3 + 11^3 + 14^3 + 4^3 + 16^3) \; \Mod \; 31 & = & 30 \\
(5^3 + 10^3 + 12^3 + 13^3 + 14^3 + 8^3) \; \Mod \; 31 & = & 29 \; ,
\end{eqnarray*}
and, so,
\begin{eqnarray*}
\encoder_2(A')
& = &
\hat{\varphi}_2 (\varphi_2(\encoder_1(A'))) \\
& = &
\left(
\arraycolsep0.2ex
\begin{array}{cccccccccc@{\;}|@{\;}ccccc@{\;}|@{\;}ccccc@{\;}|@{\;}c}
  1 & 0 & 1 & 1 & 0 & 1 & 0 & 0 & 1 & 0 & 1 & 1 & 1 & 0 & 1 &
                                          0 & 0 & 0 & 0 & 1 & 1 \\
  0 & 0 & 0 & 1 & 0 & 1 & 1 & 0 & 0 & 1 & 0 & 0 & 1 & 0 & 1 &
                                          0 & 1 & 1 & 1 & 1 & 0 \\
  0 & 1 & 0 & 0 & 0 & 1 & 0 & 1 & 1 & 1 & 0 & 0 & 0 & 1 & 0 &
                                          1 & 0 & 1 & 1 & 1 & 0
\end{array}
\right)
\; .
\end{eqnarray*}
\qed
\end{example}

It follows from the definition of~$\encoder_2$ that
every codeword $\bldc = \bldu A$ in the code~$\code_2$
induced by~$\encoder_2$ satisfies the following congruences:
\begin{eqnarray}
\label{eq:twoerror1}
\sum_{j \in \Int{n_1}} c_j \alpha_j
& \equiv & 0
\quad (\mod \; p) \\
\label{eq:twoerror2}
\sum_{j \in \Int{n_1}} c_j \alpha_j^3
& \equiv & \sum_{j \in \Int{m}} c_{n_1+j} q^j
\quad (\mod \; p) \\
\label{eq:twoerror3}
\sum_{j \in \Int{m+1}} c_j
& \equiv & 0
\quad (\mod \; 2)
\; .
\end{eqnarray}

\begin{proposition}
\label{prop:twoerrors}
The induced code~$\code_2$ satisfies
\[
\distance_\Lee(\code_2) \ge 5 \; .
\]
\end{proposition}

\begin{proof}
Let $\bldy = \bldc + \blde$ be
the read vector at the DPE output, where $\bldc \in \code$
and $\| \blde \| \le 2$.
Write $\blde = (\blde_1 \;|\; \blde_2)$ and
$\bldy = (\bldy_1 \;|\; \bldy_2)$,
where~$\blde_1$ (respectively, $\bldy_1$) is
the $n_1$-prefix of~$\blde$ (respectively, $\bldy$).
We associate with~$\bldy$ the integer syndrome
$(s_1 \; s_2 \; \hat{s}_2)$ computed as
in Eqs.~(\ref{eq:twoerrorsyndrome1})--(\ref{eq:twoerrorsyndrome3})
in Figure~\ref{fig:twoerrordecoder}.

We distinguish between the following cases.

\emph{Case 1: $s_1 = 0$.}
In this case, $\blde_1 = \bldzero$
(i.e., $\bldy_1$ is error-free),
since $\| \blde_1 \| \in \{ 1, 2 \}$ implies $s_1 \ne 0$.

\emph{Case 2: $s_1 \ne 0$ and $\hat{s}_2 = 0$.}
In this case, $\| \blde_1 \| \in \{ 1, 2 \}$,
and (by~(\ref{eq:twoerror3})),
$\bldy_2$ contains an even number of errors,
which means that $\blde_2 = \bldzero$. Therefore,
\begin{eqnarray*}
s_2 & \equiv &
\sum_{j \in \Int{n_1}} y_j \alpha_j^3
- \sum_{j \in \Int{m}} y_{n_1+j} q^j \\
& \equiv &
\sum_{j \in \Int{n_1}} y_j \alpha_j^3
- \sum_{j \in \Int{m}} c_{n_1+j} q^j \\
& \equiv &
\sum_{j \in \Int{n_1}} (c_j + e_j) \alpha_j^3
- \sum_{j \in \Int{m}} c_{n_1+j} q^j \\
& \stackrel{(\ref{eq:twoerror2})}{\equiv} &
\sum_{j \in \Int{n_1}} e_j \alpha_j^3
\quad (\mod \; p) \; .
\end{eqnarray*}
On the other hand, from~(\ref{eq:twoerror1}) we also have
\[
s_1 \equiv
\sum_{j \in \Int{n_1}} e_j \alpha_j \quad (\mod \; p) \; .
\]
Hence, we can recover~$\blde_1$ by applying a decoder
$\Decoder_\Ber$ for $\Code_\Ber(\bldalpha,2)$ to
the syndrome\footnote{%
When applying this decoder, we regard
$\bldalpha$ and $(s_1 \; s_2)$ as vectors over $\Phi = F = \GF(p)$.
The decoder~$\Decoder_\Ber$ produces
an error vector in $F^{n_1}$, which is mapped back to $\Integers^{n_1}$
by changing each given entry $z \in F$ into $\pm |z|$, with the negative
sign taken when~$z$ is a ``negative'' element of~$F$.}
$(s_1 \; s_2)$.

\emph{Case 3: $s_1 \ne 0$ and $\hat{s}_2 = 1$.}
This is possible only if
$\| \blde_1 \| = \| \blde_2 \| = 1$,
which means that we are able to recover~$\blde_1$ from~$s_1$
(using the decoding mapping in Figure~\ref{fig:oneerrordecoder}).
\end{proof}

Note that in case~1 in the last proof, $\bldy_2$ may contain
one or two errors, yet we do not attempt to decode them;
in fact, their decoding might not be unique. However,
$\bldy_1$ still decodes correctly.

The proof of Proposition~\ref{prop:twoerrors}
implies the decoding mapping~$\decoder_2$
shown in Figure~\ref{fig:twoerrordecoder}.
\begin{figure}[hbt]
\centering
\begin{myalgorithm}
\vspace{-1ex}

\quad \textbf{Input:}
$\bldy = (\bldy_1 \;|\; \bldy_2) \in \Sigma_q^n$.

\quad \textbf{Output:}
$\bldw \in \Sigma_q^k$.

\quad $//$
$n_1 = (p-1)/2$,
$m = \lceil \log_q p \rceil$,
$n_2 = n_1 + m$.

\quad $//$
$n = n_2 + 1$, $k = n_1-m$.

\quad $//$
$\bldalpha$ satisfies conditions~(i)--(iii).

\quad $//$
$\bldy_1$ is the $n_1$-prefix of~$\bldy$.

\quad $//$
$\Decoder_\Ber(\cdot)$ is a decoder for $\Code_\Ber(\bldalpha,2)$.

\begin{clause}
\item
Let
\begin{eqnarray}
\label{eq:twoerrorsyndrome1}
s_1 & \leftarrow &
\Bigl( \sum_{j \in \Int{n_1}} y_j \alpha_j \Bigr)
\; \Mod \; p \\
\label{eq:twoerrorsyndrome2}
s_2 & \leftarrow &
\Bigl( \sum_{j \in \Int{n_1}} y_j \alpha_j^3 -
\sum_{j \in \Int{m}} y_{n_1+j} q^j \Bigr)
\; \Mod \; p \\
\label{eq:twoerrorsyndrome3}
\hat{s}_2 & \leftarrow &
\Bigl( \sum_{j \in \Int{m+1}} y_{n_1+j} \Bigr)
\; \Mod \; 2  \; ;
\end{eqnarray}

\item
If $s_1 = 0$ then $\{$
\begin{clause}
\item
Let $\bldw \leftarrow \bldy'$;
\quad $//$ $\bldy_1$ is error-free
\end{clause}
$\}$

\item
Else if $\hat{s}_2 = 0$ then $\{$
\begin{clause}
\item
Let $\blde_1 \leftarrow \Decoder_\Ber(s_1 \; s_2)$;
\item
Let $\bldw \leftarrow (\bldy_1 - \blde_1)'$;
\quad $//$ $\bldw$ is the $k$-prefix of $\bldy_1 - \blde_1$
\end{clause}
$\}$

\item
Else $\{$
\begin{clause}
\item
Let $\bldw \leftarrow \decoder_1(\bldy_1)$.
\end{clause}
$\}$
\end{clause}

\end{myalgorithm}
\caption{Decoding mapping $\decoder_2 : \bldy \mapsto \bldw$
for double-error correction.}
\label{fig:twoerrordecoder}
\end{figure}

We include in Appendix~\ref{sec:example}
an example of an application of
the decoder in Figure~\ref{fig:twoerrordecoder};
for self-containment, that example also recalls
the decoding principles of Berlekamp codes.

The coding scheme $(\encoder_2,\decoder_2)$ can be extended
to also detect three errors, by adding an overall parity bit
to each row of the encoded matrix~$A$, as was done
in Subsection~\ref{sec:doubleerrordetection}. Moreover,
the savings shown there when $q > 2$ carries over also
to minimum distances~$5$ and~$6$.

Specifically, for odd~$q$, we redefine~$m$ to be
$\lceil \log_q (2p) \rceil = \lceil \log_q (4n_1+2) \rceil$,
and require\footnote{%
See Footnote~\ref{fn:violationqodd} for the case
where $q = p$.}
$\bldalpha$ to satisfy conditions~(i')--(ii').
The encoding mapping~$\encoder_2$ is redefined to just
$\varphi_2 \circ \encoder_1$, with code length
$n = n_2 = n_1 + m$ and redundancy $n-k = 2m$.
The component functions~$\encoder_1$ and~$\varphi_2$
are as in Figure~\ref{fig:twoerrorencoder}, except that
all remainders modulo~$p$ are now computed modulo~$2p$.

The function of the syndrome element~$\hat{s}_2$ in the decoding
process is replaced by the parities of~$s_1$ and~$s_2$,
when computed as in~(\ref{eq:twoerrorsyndrome1})
and~(\ref{eq:twoerrorsyndrome2}),
except that the remainders are taken modulo~$2p$
(yet~$\Code_\Ber$ is still defined over $F = \GF(p)$,
so when its decoder is applied to $(s_1 \; s_2)$,
the syndrome entries are reduced first modulo~$p$).
Specifically, assuming that
at most three errors have occurred,
Table~\ref{tab:threeerrors} presents the various
combinations of parities of~$s_1$ and~$s_2$,
and the corresponding $L_1$-norms of~$\blde_1$ and~$\blde_2$.
\begin{table}[hbt]
\caption{Decoding two errors and detecting three errors.}
\label{tab:threeerrors}
\centering
\tabcolsep1.0ex
\begin{tabular}{cccccc}
\hline\hline
$s_1$ & $s_2$ & $s_2 \equiv s_1^3$? & $\| \blde_1 \|$ & $\| \blde_2 \|$
                                    & Decoder output                  \\
\hline\hline
$0$          &$-$  &$-$ &$0$&$-$& $\bldy'_1$                          \\
even $\ne 0$ &even &$-$ &$2$&$0$& $(\bldy_1-\Decoder_\Ber(s_1\;s_2))'$\\
odd          &even &$-$ &$1$&$1$& $\decoder_1(\bldy_1)$               \\
even $\ne 0$ &odd  &$-$ &$2$&$1$& $\failure$                          \\
odd          &odd  &yes &$1$&$0$& $\decoder_1(\bldy_1)$               \\
odd          &odd  &no  &$3$&$0$& $\failure$                          \\
odd          &odd  &no  &$1$&$2$& $\failure$                          \\
\hline\hline
\end{tabular}
\end{table}
The first three rows in the table correspond, respectively, to the three
cases in the proof of Proposition~\ref{prop:twoerrors}.
The fourth row corresponds to three errors and therefore
should result in a decoding failure. The last three rows
correspond to three different combinations of number of errors:
the distinction among them can be made
by checking  whether $s_2 \equiv s_1^3 \; (\mod \; p)$,
and, since $\Code_\Ber(\bldalpha,2)$ has minimum Lee distance~$5$,
this condition will be met only when one error has occurred
(in which case it can be found using
the decoding mapping in Figure~\ref{fig:oneerrordecoder}).

When~$q$ is even and greater than~$2$, we will follow the same
strategy as in Section~\ref{sec:oneerror}, namely,
replacing the terms\footnote{%
See Footnote~\ref{fn:violationqeven} for the case
where $f_{m-1}(q) = p$.}
$q^j$ with $f_j(q)$ defined in~(\ref{eq:fjq}),
both in condition~(iii) and in~(\ref{eq:varphi2}).

\begin{example}
\label{ex:evenq>2twoerrors}
Suppose that $q = 4$ and $p = 101$, corresponding to
$n_1 = 50$. The values $f_j(4)$ for $j = 0, 1, 2, 3, 4$
are~$1$, $3$, $13$, $51$, and~$205$, respectively,
so we can take $m = 4$ and
\[
\bldalpha =
\left(
\, 5 \; 7 \; 9 \; 11 \; 15 \; 17 \, \ldots \, 47 \; 49 \; 53 \; 55
\, \ldots \,
97 \; 99
\bigm|
1 \; 3 \; 13 \; 51 \, \right)
.
\]
The respective double-error-correcting triple-error-detecting
coding scheme has length $n = n_1 + m = 54$ and redundancy
$n - k = 2m = 8$ (and dimension $k = 46$).
An example of an image of the encoding mapping
$\encoder_2 = \varphi_2 \circ \encoder_1$ is given by
\[
A = \encoder_2(A') =
\left(
\arraycolsep0.2ex
\begin{array}{cccccccccc@{\;}|@{\;}cccc@{\;}|@{\;}cccc}
1& 2& 3& 0& 1& 2& 0& 0& \ldots & 0& 2& 1& 0& 2& 0& 1& 3& 2 \\
0& 3& 0& 1& 2& 3& 0& 0& \ldots & 0& 3& 3& 2& 1& 1& 2& 2& 3 \\
2& 1& 1& 3& 2& 0& 0& 0& \ldots & 0& 2& 3& 0& 2& 2& 0& 0& 2
\end{array}
\right)
\; ,
\]
with
\[
\bldc = \bldu A =
\left( \,
4\; 14\; 7\; 6\; 10\; 13\; 0\; 0 \, \ldots \, 0 \bigm|
15\; 14\; 6\; 9 \bigm| 5\; 8\; 12\; 15
\, \right)
\]
being an example of a codeword (which corresponds to the input
vector $\bldu = ( 2 \; 3 \; 1 )$).

Given a read vector
$\bldy = (y_j)_{j \in \Int{n}} \in \Sigma_Q^n$
(where $Q = 28$ and $n = 54$), its syndrome is given by
\begin{eqnarray*}
s_1 & = &
\Bigl( \sum_{j \in \Int{n_1}} y_j \alpha_j \Bigr)
\; \Mod \; (2p) \\
s_2 & = &
\Bigl( \sum_{j \in \Int{n_1}} y_j \alpha_j^3 -
\sum_{j \in \Int{m}} y_{n_1+j} f_j(q) \Bigr)
\; \Mod \; (2p)
\end{eqnarray*}
(where $q = 4$, $p = 101$, $n_1 = 50$, and $m = 4$).
Correction of two errors and detecting of three then
proceeds by following Table~\ref{tab:threeerrors},
where~$\bldy_1$ and~$\bldy'_1$ are the prefixes of~$\bldy$
of lengths $n_1 = 50$ and $k = 46$, respectively.\qed
\end{example}

\subsection{Recursive coding scheme}
\label{sec:recursive}

The construction in Subsection~\ref{sec:twoerrors} does not seem
to generalize in a straightforward way to larger
minimum $L_1$-distances. However, with some redundancy increase
(which will be relatively mild for code lengths sufficiently large),
we can construct coding schemes for
any prescribed number of correctable errors. We show this next.

Given the alphabet~$\Sigma_q$, number of rows~$\ell$,
designed number of correctable errors\footnote{%
For simplicity, we assume that $\sigma = 0$, namely,
no additional errors are to be detected.}
$\tau$, let $p > 2\tau$ be a prime, and define $n = (p-1)/2$
and $m = \lceil \log_q p \rceil$. Also, let
$\bldalpha = (\alpha_j)_{j \in \Int{n}}$ be an integer vector
that satisfies conditions~(i)--(iii) in Section~\ref{sec:oneerror}.

Given a matrix $A \in \Sigma_q^{\ell \times n}$
(which, at this point, is not assumed to be the result of any encoding),
we can compute
the following $\ell \times \tau$ \emph{syndrome matrix} of~$A$ over
$\Integers$:
\[
S = (s_{i,v})_{i \in \Int{\ell},v \in \Int{\tau}}
= A H_\Ber^\transpose \; \Mod \; p \; ,
\]
where $H_\Ber = H_\Ber(\bldalpha,\tau)$ is the parity-check matrix
defined in~(\ref{eq:HBer}), now seen as a matrix over~$\Integers$,
and the remainder computed entry-wise.
For a vector $\bldu \in \Sigma_q^\ell$, the syndrome
$\blds = \blds(\bldu)$ of $\bldc = \bldc(\bldu) = \bldu A$ is
then given by
\begin{eqnarray*}
\blds & = & \bldc H_\Ber^\transpose \; \Mod \; p \\
      & = & \bldu A H_\Ber^\transpose \; \Mod \; p \\
      & = & \bldu S \; \Mod \; p \; .
\end{eqnarray*}
If the syndrome~$\blds$ is available to the decoder,
then the decoder should be able to recover $\bldc = \bldu A$ from
an erroneous copy $\bldy = \bldc + \blde \; (\in \Integers^n)$,
provided that $\| \blde \| \le \tau$: this is simply because
the syndrome~$\hat{\blds}$ of~$\blde$ is computable from~$\blds$
and the syndrome of~$\bldy$,
\begin{eqnarray*}
\hat{\blds} =
\blde H_\Ber^\transpose \; \Mod \; p
      & = & (\bldy - \bldc) H_\Ber^\transpose \; \Mod \; p \\
      & = & (\bldy H_\Ber^\transpose  - \blds) \; \Mod \; p \; ,
\end{eqnarray*}
and $\blde \leftarrow \Decoder_\Ber(\hat{\blds})$,
where $\Decoder_\Ber(\cdot)$ is a decoder
for $\Code_\Ber(\bldalpha,\tau)$.
Thus, our encoding mapping will be designed so that,
\emph{inter alia}, the decoder is able to reconstruct a copy of~$\blds$.

Each entry in~$S$, being an integer in $\Int{p}$, can be
expanded to its base-$q$ representation
\[
s_{i,v} = \sum_{j \in \Int{m}} s_{i,v}^{(j)} q^j \; ,
\]
where $s_{i,v}^{(j)} \in \Sigma_q$.
For $j \in \Int{m}$, let~$S^{(j)}$ be the $\ell \times \tau$ matrix
$(s_{i,v}^{(j)})_{i \in \Int{\ell},v \in \Int{\tau}}$ over~$\Sigma_q$.
Clearly,
\[
S = \sum_{j \in \Int{m}} q^j S^{(j)}
\]
and, so,
\begin{eqnarray*}
\blds & = & \bldu S \; \Mod \; p \\
      & = &
\Bigl( \sum_{j \in \Int{m}} q^j (\bldu S^{(j)}) \Bigr) \; \Mod \; p \\
      & = &
\Bigl( \sum_{j \in \Int{m}} q^j \blds^{(j)} \Bigr) \; \Mod \; p \; ,
\end{eqnarray*}
where $\blds^{(j)} = \blds^{(j)}(\bldu) = \bldu S^{(j)}$ is
a vector in~$\Sigma_Q^m$. Consider an encoding mapping
$\encoder : \Sigma_q^{\ell \times n} \rightarrow
\Sigma_q^{\ell \times (n + \tau m)}$ defined by
\[
\encoder : A \mapsto
\left( A \;|\; S^{(0)} \;|\; S^{(1)} \;|\; \cdots \;|\;
S^{(m-1)} \right)
\; .
\]
Then, for $\bldu \in \Sigma_q^\ell$ we have
\[
\bldu \, \encoder(A) =
\left( \bldc \;|\; \blds^{(0)} \;|\; \blds^{(1)} \;|\; \cdots \;|\;
\blds^{(m-1)} \right)
\]
(where $\bldc = \bldc(\bldu)$ and $\blds^{(j)} = \blds^{(j)}(\bldu)$).
If $\bldy = \bldu \, \encoder(A) + \blde$ where $\| \blde \| \le \tau$,
then, based on our previous discussion,
we will be able to recover~$\bldc$,
\emph{as long as the $\tau m$-suffix of~$\bldy$ is error-free}.

The latter assumption (of an error-free suffix) can be guaranteed
by applying a (second) encoding mapping to
the $\ell \times \tau m$ matrix
\[
\left( S^{(0)} \;|\; S^{(1)} \;|\; \cdots \;|\; S^{(m-1)} \right)
\]
(over~$\Sigma_q$) so that~$\tau$ errors can be corrected.
Note that the matrix now has $\tn = \tau m$ columns
(instead of~$n$), so we can base our encoding on
a Berlekamp code over $\GF(\tilde{p})$,
where~$\tilde{p}$ is the smallest prime which is at least
$2 \tn + 1$.
The size of the syndrome now will be $\tau \tilde{m}$, where
$\tilde{m} = \lceil \log_q \tilde{p} \rceil$, namely, becoming
\emph{doubly-logarithmic} in~$n$.

We can continue this process recursively; by just applying
one more recursion level with a simple repetition encoding mapping
(which copies its input $2\tau+1$ times at the output), we obtain
a total redundancy of
\begin{eqnarray}
\lefteqn{
\tau m + (2\tau + 1) \tau \tilde{m} =
\tau \lceil \log_q (2n+1) \rceil
} \makebox[15ex]{} \nonumber \\
\label{eq:redundancy}
&&
\quad {} +
O \left(\tau^2 \log_q (\tau \log_q n) \right) \; .
\end{eqnarray}
Hence, for~$n$ large compared to~$\tau$, most of
the redundancy is due to the first encoding level.
In fact, by extending the sphere-packing argument
presented at the end of Subsection~\ref{sec:oneerror}
it follows that the redundancy~(\ref{eq:redundancy}) is optimal,
up to an additive term that depends on~$\tau$, but not on~$n$.

Decoding is carried out backwards, starting with recovering
the codeword that corresponds to the last encoding level, which,
in turn, serves as the syndrome of the previous encoding level.

Reflecting now back on our constructions in
Sections~\ref{sec:oneerror} and~\ref{sec:twoerrors},
if the matrix $A \; (\in \Sigma_q^{\ell \times n})$ is the output
of the encoding scheme in Figure~\ref{fig:oneerrorencoder},
then the first column of the syndrome matrix~$S$ is zero,
and therefore so is the first column in each matrix $S^{(j)}$
(and the first entry in each vector $\blds^{(j)}$). Hence, those
zero columns can of course be removed. As for the second column,
our construction in Subsection~\ref{sec:twoerrors} implies
that it can be error-protected simply by a parity bit
(or, when $q > 2$, by changing the modulus from~$p$ to~$2p$
and selecting the entries of~$\bldalpha$ to be odd).

The approach of recursive encoding is not new, and
has been used, for example, in the context of constrained coding
(e.g., see~\cite{Bliss}, \cite{FC}, \cite{Immink}, \cite{Knuth},
\cite{Mansuripur}).
In our setting, this approach
allows us to use codes (namely, Berlekamp codes), which are originally
defined over one alphabet of size~$p$, while the result of
the encoding (namely, the contents of the rows of the DPE matrix)
are restricted to belong to another alphabet of size~$q$
(the challenge is evident when $q < p$). In the next subsection,
we consider a more straightforward application of Berlekamp codes
to construct a coding scheme for the case where~$q$ is large enough;
this scheme may sometimes have a smaller redundancy
than~(\ref{eq:redundancy}).

\subsection{Coding scheme for large alphabets}
\label{sec:largealphabets}

We consider here the case where the number of correctable errors
$\tau$ and the alphabet size~$q$ are such that there exists
a prime~$p$ that satisfies
\[
2\tau < p \le q \; .
\]
We will then assume that~$p$ is the largest prime that does
not exceed~$q$, and we let~$F$ be the finite field $\GF(p)$.

We will use a systematic encoder
$\Encoder_\Ber : F^k \rightarrow \Code_\Ber$,
where $\Code_\Ber = \Code_\Ber(\bldbeta,\tau)$ is a Berlekamp code of
a prescribed length~$n$ over~$F$ and redundancy
\begin{eqnarray*}
n-k
& \le & \tau \cdot \lceil \log_p (2n+1) \rceil \\
& = &
\tau \cdot \lceil (\log_p q) \cdot \log_q (2n+1) \rceil
\end{eqnarray*}
(when~$n$ is sufficiently large compared to~$\tau$,
the inequality is known to hold with equality).
When $q = p$, this redundancy is smaller
than~(\ref{eq:redundancy}); otherwise, it will be larger
for~$\tau$ (much) smaller than~$n$, due to the factor $\log_p q$
(e.g., for $q = 8$, this factor is approximately $1.07$).

Our encoding mapping
$\encoder:\Sigma_q^{\ell \times k} \rightarrow \Sigma_q^{\ell \times n}$
takes each row in the pre-image matrix
$A' \in \Sigma_q^{\ell \times k}$, computes the remainder
of each entry modulo~$p$, regards the result as a vector
in $F^k$, and applies to it the encoder~$\Encoder_\Ber$ to produce
a codeword $\bldc \in \Code_\Ber$.
The $(n{-}k)$-suffix, $\bldc''$, of~$\bldc$ becomes
the $(n{-}k)$-suffix of the respective row in
the image $A = (A' \;|\; A'') = \encoder(A')$
(see Figure~\ref{fig:largealphabetencoder}).

\begin{figure}[hbt]
\centering
\begin{myalgorithm}
\vspace{-1ex}

\quad \textbf{Input:}
$\ell \times k$ matrix
$A' = (a_{i,j})_{i \in \Int{\ell}, j \in \Int{k}}$ over~$\Sigma_q$.

\quad \textbf{Output:}
$\ell \times n$ matrix
$(A' \;|\; A'') = (a_{i,j})_{i \in \Int{\ell}, j \in \Int{n}}$
over~$\Sigma_q$.

\quad $//$
$F = \GF(p)$, for a prime~$p$ such that
$2\tau + 1 \le p \le q$.

\quad $//$
$\Code_\Ber(\bldbeta,\tau)$ is a Berlekamp code of length~$n$ over~$F$.

\quad $//$
$\Encoder_\Ber : F^k \rightarrow \Code_\Ber(\bldbeta,\tau)$
is a systematic encoder.

\begin{clause}
\item
For all $i \in \Int{\ell}$ do $\{$
\begin{clause}
\item
Let $\bldc \leftarrow \Encoder_\Ber(A'_i \; \Mod \; p)$;
\item
Let $A''_i \leftarrow \bldc''$.
\quad $//$ $F$ seen as a subset of~$\Sigma_q$
\end{clause}
$\}$
\end{clause}

\end{myalgorithm}
\caption{Encoding mapping
$\encoder : A' \mapsto (A' \;|\; A'')$
for large alphabets.}
\label{fig:largealphabetencoder}
\end{figure}

It follows from the construction that the codewords
of the induced code~$\code$, when reduced modulo~$p$,
are codewords of~$\Code_\Ber$. Thus, we obtain a coding scheme
that can correct~$\tau$ errors.

\section{Coding scheme for the Hamming metric}
\label{sec:Hamming}

In this section, we present a coding scheme that handles
errors in the Hamming metric; namely, the number of errors
equals the number of positions in which the read vector
$\bldy = (y_j)_{j \in \Int{n}} \in \Sigma_Q^n$
differs from the correct computation
$\bldc = (c_j)_{j \in \Int{n}} = \bldu A$.

For the purpose of the exposition, we will introduce
yet another design parameter, $\threshold$, which will be
an assumed upper bound on the absolute value of any error value,
namely, on
\[
\max_{j \in \Int{n}} |y_j - c_j| \; .
\]
Such an error model may be of independent interest in DPE applications,
with the special case $\threshold = Q-1 = \ell(q{-}1)^2$
being equivalent to the ordinary Hamming metric.

Given the alphabet~$\Sigma_q$, number of rows~$\ell$,
number of columns~$n$, number of correctable errors\footnote{%
We assume that $\sigma = 0$ (as in Subsection~\ref{sec:recursive})
and that there are no erasures.}
$\tau$,
and an upper bound~$\threshold$ on the error absolute value
(where $\threshold = \ell(q{-}1)^2$ for unconstrained error values),
let~$p$ be a prime greater than $2 \threshold$ and
let $m = \lceil \log_q p \rceil$.
We select
a respective linear $\tau$-error-correcting $[\tn,k]$ code~$\Code$
over $F = \GF(p)$ (in the Hamming metric), which is assumed to have
an efficient bounded-distance decoder
$\Decoder : F^\tn \rightarrow F^\tn$:
for a received word $\tilde{\bldy} \in F^\tn$, the decoder returns
the true error vector $\tilde{\blde} \in F^\tn$,
provided that its Hamming weight $\weight(\tilde{\blde})$
was at most~$\tau$.

The parameters~$n$ and~$\tn$ are related by
\[
n = k + m(\tn-k) \; .
\]
Figure~\ref{fig:Hammingencoder} presents
an encoding mapping $\encoder : A' \mapsto (A' \;|\; A'')$,
where each row of~$A'$, when reduced modulo~$p$, is first
extended by the systematic encoder for~$\Code$ into a codeword
$\tilde{\bldc}$ of~$\Code$, and then
the $\tn-k$ redundancy symbols (over~$F$) in
$\tilde{\bldc}$ are expanded to their base-$q$ representations
to form the respective row in~$A''$. Specifically,
\[
A'' = \left( A^{(0)} \;|\; A^{(1)} \;|\; \cdots \;|\; A^{(m-1)} \right)
\; ,
\]
where each block $A^{(j)}$ is
an $\ell \times (\tn{-}k)$ sub-matrix
over~$\Sigma_q$, such that the rows of
the $\ell \times \tn$ matrix
\[
\tilde{A} =
\left( \textstyle A' \;\Bigm|\; \sum_{j \in \Int{m}} q^j A^{(j)} \right)
\; \Mod \; p
\]
form codewords of~$\Code$.
\begin{figure}[hbt]
\centering
\begin{myalgorithm}
\vspace{-1ex}

\quad \textbf{Input:}
$\ell \times k$ matrix
$A' = (a_{i,j})_{i \in \Int{\ell}, j \in \Int{k}}$ over~$\Sigma_q$.

\quad \textbf{Output:}
$\ell \times n$ matrix
$(A' \;|\; A'') = (a_{i,j})_{i \in \Int{\ell}, j \in \Int{n}}$
over~$\Sigma_q$.

\quad $//$
$F = \GF(p)$, for a prime $p > 2 \threshold$.

\quad $//$
$m = \lceil \log_q p \rceil$
and $n = k + m(\tn-k)$.

\quad $//$
$\Code$ is a linear $\tau$-error-correcting $[\tn,k]$ code over~$F$.

\quad $//$
$\Encoder : F^k \rightarrow \Code$ is a systematic encoder.

\begin{clause}
\item
For all $i \in \Int{\ell}$ do $\{$
\begin{clause}
\item
Let $\tilde{\bldc} = (\tilde{c}_v)_{v \in \Int{\tn}}
\leftarrow \Encoder(A'_i \; \Mod \; p)$;
\item
For each $v \in \Int{\tn{-}k}$ do $\{$
\begin{clause}
\item
Set
$((A^{(0)})_{i,v} \; (A^{(1)})_{i,v} \; \ldots \; (A^{(m-1)})_{i,v})$
to be
\item
\quad
the base-$q$ representation of $\tilde{c}_{k+v}$.
\end{clause}
$\}$
\end{clause}
$\}$

\item
Let
$A'' \leftarrow
\left( A^{(0)} \;|\; A^{(1)} \;|\; \cdots \;|\; A^{(m-1)} \right)$.
\end{clause}

\end{myalgorithm}
\caption{Encoding mapping $\encoder : A' \mapsto (A' \;|\; A'')$
for the Hamming metric.}
\label{fig:Hammingencoder}
\end{figure}

Let the mapping $\lambda : \Integers^n \rightarrow F^{\tn}$ be
defined as follows:
for a vector $\bldx = (x_v)_{v \in \Int{n}} \in \Integers^n$,
the entries of the image
$\lambda(\bldx)=\tilde{\bldx}=(\tilde{x}_v)_{v \in \Int{\tn}} \in F^\tn$
are given by
\begin{equation}
\label{eq:homomorphism1}
\tilde{x}_v = x_v \; \Mod \; p \; ,
\quad
\textrm{for $v \in \Int{k}$} \; ,
\end{equation}
and
\begin{equation}
\label{eq:homomorphism2}
\tilde{x}_{k+v} =
\Bigl( \sum_{j \in \Int{m}} x_{k+v+j(\tilde{n}-k)}
q^j \Bigr) \; \Mod \; p \; ,
\quad
\textrm{for $v \in \Int{\tn{-}k}$} \; .
\end{equation}
It is easy to see that each row in~$\tilde{A}$ is obtained
by applying the mapping~$\lambda$ to the respective row
in~$A = (A' \;|\; A'')$.
Moreover, $\lambda$ is a homomorphism in that
it preserves vector addition and scalar multiplication: for every
$\bldx_1, \bldx_2 \in \Integers^n$ and $b_1, b_2 \in \Integers$,
\[
\lambda(b_1 \bldx_1 + b_2 \bldx_2) =
\overline{b}_1 \lambda(\bldx_1)
+ \overline{b}_2 \lambda(\bldx_2) \; ,
\]
where $\overline{b}_i = b_i \; \Mod \; p$
(seen as elements of~$F$).
Consequently, for every $\bldu \in \Sigma_q^k$,
\[
\lambda(\bldu A) = \bldu \tilde{A} \; \Mod \; p \in \Code \; .
\]

The properties of $\lambda(\cdot)$ immediately imply
a decoding algorithm (shown in Figure~\ref{fig:Hammingdecoder}).
Given the read vector $\bldy = \bldc + \blde$,
where $\weight(\blde) \le \tau$,
an application of~$\lambda$ to~$\bldy$ yields:
\[
\lambda(\bldy) = \lambda(\bldc) + \lambda(\blde) \; ,
\]
where $\lambda(\bldc) \in \Code$ and
$\weight(\lambda(\blde)) \le \weight(\blde) \le \tau$.
Hence, a decoder for~$\Code$, when applied
to~$\lambda(\bldy)$, will recover $\lambda(\blde)$.
By the definition of~$\lambda$,
the vectors~$\blde$ and $\lambda(\blde)$ coincide, modulo~$p$,
on their $k$-prefix; and since the values of the entries of
$\blde$ are all within $\pm \threshold$,
the $k$-prefix of $\lambda(\blde)$ uniquely determines
the $k$-prefix of~$\blde$.
\begin{figure}[hbt]
\centering
\begin{myalgorithm}
\vspace{-1ex}

\quad \textbf{Input:}
$\bldy \in \Sigma_q^n$.

\quad \textbf{Output:}
$\bldw \in \Sigma_q^k$.

\quad $//$
Parameters are as defined in Figure~\protect\ref{fig:Hammingencoder}.

\quad $//$
$\lambda : \Integers^n \rightarrow F^{\tn}$ is defined
by~(\ref{eq:homomorphism1})--(\ref{eq:homomorphism2}).

\quad $//$
$\Decoder : F^\tn \rightarrow F^\tn$ is a decoder for~$\Code$.

\begin{clause}
\item
Let $\tilde{\blde} = (\tilde{e}_j)_{j \in \tn}
\leftarrow \Decoder(\lambda(\bldy))$;

\item
Set $\blde' = (e_j)_{j \in \Int{k}}$ to
\[
e_j \leftarrow
\left\{
\begin{array}{rcl}
|\tilde{e}_j|  && \textrm{if $\tilde{e}_j$ is ``nonnegative'' in~$F$} \\
-|\tilde{e}_j| && \textrm{otherwise}
\end{array}
\right.
;
\]

\item
Let $\bldw \leftarrow \bldy' - \blde'$.
\end{clause}

\end{myalgorithm}
\caption{Decoding mapping
$\decoder : \bldy \mapsto \bldw$ for the Hamming metric.}
\label{fig:Hammingdecoder}
\end{figure}

Finally, we specialize to the case where~$\Code$ is
a (normalized and possibly shortened) BCH code. In this case,
\[
\tn-k \le
\left\lceil 1 + \frac{p{-}1}{p} (2\tau{-}1)
\right\rceil \cdot \left\lceil \log_p \tn \right\rceil
\]
(see~\cite[p.~260, Problem~8.13]{Roth}), and, so,
the redundancy of our coding scheme is bounded from above by
\[
n-k \le
\left\lceil 1 + \frac{p{-}1}{p} (2\tau{-}1)
\right\rceil \cdot
\underbrace{
\left\lceil \log_p \tn \right\rceil
\cdot \left\lceil \log_q p \right\rceil
}_{\approx \, \log_q n}
\; .
\]
For reference, recall that for every row index $i \in \Int{\ell}$,
the possible contents of~$A_i$ must form
an (ordinary) code over~$\Sigma_q$
of minimum Hamming distance at least~$2\tau+1$
(assuming that $\threshold \ge q{-}1$).
For~$n$ sufficiently large compared to~$\tau$, BCH codes over
$\GF(q)$ are the best codes currently known for all prime powers~$q$
except~$4$ and~$8$. Hence, we should expect the redundancy
of the coding scheme to be no less than
\[
\left\lceil 1 + \frac{q{-}1}{q} (2\tau{-}1)
\right\rceil \cdot
\left\lceil \log_q n \right\rceil \; .
\]

\section*{Acknowledgment}
The author would like to thank Dick Henze, Naveen Muralimanohar,
and John Paul Strachan for introducing him to the problem,
and for the many helpful discussions.

\ifIEEE
   \appendices
\else
   \section*{$\,$\hfill Appendices\hfill$\,$}
   \appendix
\fi

\section{Proofs}
\label{sec:proofs}

\begin{proof}[Proof of Proposition~\ref{prop:distance}]
For any two codewords $\bldc_1, \bldc_2 \in \code$
and a vector $\bldy \in \Sigma_Q^n$, we have
\[
\distance(\bldc_1,\bldc_2) \le
\distance(\bldy,\bldc_1) + \distance(\bldy,\bldc_2) \; .
\]
Hence, the inequalities
\[
\distance(\bldy,\bldc_1) \le \tau
\quad \textrm{and} \quad
\distance(\bldy,\bldc_2) \le \tau + \sigma
\]
can hold simultaneously, only if
$\distance(\bldc_1,\bldc_2) \le 2 \tau + \sigma < \distance(\code)$,
namely, only when $\bldc_1' = \bldc_2'$.
This, in turn, implies that the following decoding mapping is
well-defined and satisfies the correction and detection
conditions above: for every $\bldy \in \Sigma_Q^n$,
\[
\decoder(\bldy) =
\left\{
\begin{array}{ccl}
\bldc'      &&
\textrm{if there is $\bldc \in \code$ such that
$\distance(\bldy,\bldc) \le \tau$} \\
\failure && \textrm{otherwise}
\end{array}
\right.
.
\]
\end{proof}

(It can be easily shown that the condition on~$\tau$ and~$\sigma$
in Proposition~\ref{prop:distance} is also necessary:
if $2 \tau + \sigma \ge \distance(\code)$, then there can be
no decoding mapping that corrects
(\emph{any} pattern of up to) $\tau$ errors and detects
$\tau + \sigma$ errors.)

\begin{proof}[Proof of~\ref{prop:distanceHamming}]
Let~$\bldc_1$ and~$\bldc_2$ be codewords in~$\code$
with distinct $k$-prefixes. Ignoring the coordinates
that have been erased, these codewords will still differ
on at least $\distance_\Hamming(\bldc_1,\bldc_2)-\rho$ coordinates.
Next apply Proposition~\ref{prop:distance},
with $\distance(\code)$ therein replaced by
$\distance_\Hamming(\code)-\rho$.
\end{proof}

\section{Example}
\label{sec:example}

We include here an example of an execution
of the decoder in Figure~\ref{fig:twoerrordecoder}.

\begin{example}
\label{ex:twoerrordecoder}
Continuing Example~\ref{ex:twoerrorencoder},
for $\bldu = ( 1 \; 1 \; 1 )$, we get
\[
\bldc = \bldu A =
\left( \, 1\; 1\; 1\; 2\; 0\; 3\; 1\; 1\; 2\; 2 \bigm| 1\; 1\; 2\; 1\; 2
\bigm| 2\; 0\; 1\; 2\; 1 \bigm| 2 \, \right) \; .
\]
Suppose that the read vector is
\[
\bldy =
\left( \, 1\; 1\; 1\; 2\; 0\; 2\; 1\; 1\; 2\; 2 \bigm| 1\; 1\; 2\; 2\; 2
\bigm| 2\; 0\; 1\; 2\; 1 \bigm| 2 \, \right) \; .
\]
The syndrome of~$\bldy$ is
computed by~(\ref{eq:twoerrorsyndrome1})--(\ref{eq:twoerrorsyndrome3})
to yield
\[
(s_1 \; s_2 \; \hat{s}_2) =
(29 \; 8 \; 0) \; .
\]
A nonzero~$s_1$ indicates that (one or two) errors have occurred in
the $n_1$-prefix, $\bldy_1$, of~$\bldy$, and a zero~$\hat{s}_2$
then indicates that the remaining part of~$\bldy$ is error-free.
Hence, we are in the scenario of case~2
in the proof of Proposition~\ref{prop:twoerrors}, i.e.,
$(s_1 \; s_2) = \blde_1 H_\Ber^\transpose \; \Mod \; p$
(where $p = 31$),
which allows us to find~$\blde_1$ using a decoder for~$\Code_\Ber$.
Next, we recall the principles of the decoding algorithm
for~$\Code_\Ber$, by demonstrating them on our particular example.

We first observe that if $\| \blde_1 \| = 1$
then necessarily $s_2 \equiv s_1^3 \; (\mod \; 31)$.
Since this congruence does not hold in our case, we deduce
that that two errors have occurred,
say at positions $i, j \in \Int{n_1}$. We have
\begin{eqnarray}
\label{eq:s1}
s_1 & \equiv & e_i \alpha_i + e_j \alpha_j
\quad (\mod \; p) \\
\label{eq:s2}
s_2 & \equiv & e_i \alpha_i^3 + e_j \alpha_j^3
\quad (\mod \; p) \; ,
\end{eqnarray}
where $e_i, e_j \in \{ -1, 1 \}$. Squaring both sides
of~(\ref{eq:s1}) yields
\[
s_1^2 \equiv \alpha_i^2 + 2 (e_i \alpha_i) (e_j \alpha_j) + \alpha_j^2
\quad (\mod \; p) \; ,
\]
while dividing each side of~(\ref{eq:s2}) by the respective
side in~(\ref{eq:s1}) yields
\[
\frac{s_2}{s_1} \equiv
\alpha_i^2 - (e_i \alpha_i) (e_j \alpha_j) + \alpha_j^2
\quad (\mod \; p) \; ,
\]
where $1/s_1$ stands for the inverse of~$s_1$ modulo~$p$.
Subtracting each side of the last congruence from
the respective sides of the previous congruence leads to
\begin{equation}
\label{eq:s1s2}
s_1^2 - \frac{s_2}{s_1} \equiv
3 (e_i \alpha_i) (e_j \alpha_j)
\quad (\mod \; p) \; .
\end{equation}
It follows from~(\ref{eq:s1}) and~(\ref{eq:s1s2})
that $e_i \alpha_i$ and $e_j \alpha_j$ are solutions
of the following quadratic equation (in $F = \GF(p)$):
\[
x^2 -  s_1 x +
\frac{1}{3} \left( s_1^2 - \frac{s_2}{s_1} \right)
\equiv 0 \quad (\mod \; p) \; .
\]
Specifically, in our case,
\[
\frac{1}{s_1} = \frac{1}{29} \equiv 15 \quad (\mod \; 31)
\;\;\quad \textrm{and} \;\;\quad
\frac{1}{3} \equiv 21 \quad (\mod \; 31) \; ,
\]
resulting in the quadratic equation
\[
x^2 + 2 x + 13 \equiv 0 \quad (\mod \; 31) \; .
\]
The two roots of this equation in $\GF(31)$ are~$8$ and~$21$:
the first points at an error with a value~$1$ at location~$13$
(since $\alpha_{13} = 8$), and the second points
at an error with a value~$-1$ at location~$5$
(since $\alpha_5 = 10 = -21 \; \Mod \; 31$).\qed
\end{example}

\end{document}